\documentclass[acmtog]{acmart}

\AtBeginDocument{%
  }

\setcopyright{none}
\renewcommand\footnotetextcopyrightpermission[1]{} 
\copyrightyear{2026}
\acmYear{2026}

\usepackage{svg}
\usepackage{comment}
\usepackage{version}
\usepackage[normalem]{ulem}

\usepackage{mathtools} 
\usepackage{textcomp}
\usepackage{siunitx} 
\usepackage{physics} 
\usepackage{stmaryrd}
\usepackage{scalerel}
\usepackage{bm}

\usepackage{ifthen}
\usepackage{csquotes}
\usepackage{calc}
\usepackage{longtable}

\usepackage{algorithm}
\usepackage{algorithmicx}
\usepackage{algpseudocode}
\usepackage{listingsutf8}

\usepackage{array}
\usepackage{multirow}

\usepackage{wrapfig} 
\usepackage[percent]{overpic}
\usepackage{subcaption}
\usepackage{supertabular}

\usepackage{cleveref}

\definecolor{DenisColor}{rgb}{0.55,0.35,0.05}
\definecolor{DanieleColor}{rgb}{0.3411764706,0.02352941176,0.5490196078}
\definecolor{LeyiColor}{rgb}{0.99,0.37,0.07}
\definecolor{YixinColor}{rgb}{0.6274509804,0.2039215686,0.4470588235}
\definecolor{MichaelColor}{rgb}{0.1, 0.5, 0.5}

\newcommand{\DZ}[1]{{\leavevmode\color{DenisColor} Denis: #1 $\qed$}}

\newcommand{\todo}[1]{\textcolor{red}{TODO: #1 $\qed$}}

\providecommand{\finalversion}{0} 
\ifthenelse{\equal{\finalversion}{1}}
{
	\renewcommand{\denis}[1]{}
    \renewcommand{\DZ}[1]{}
	\renewcommand{\toref}[1]{}
	\renewcommand{\tocite}[1]{}
	\renewcommand{\todo}[1]{}
	\renewcommand{\warning}[1]{}
	\renewcommand{\note}[1]{}
}
{}

\definecolor{forestgreen}{rgb}{0.13,0.54,0.13}
\definecolor{darkblue}{rgb}{0,0,.5}
\hypersetup{
	citecolor=forestgreen, 
	linkcolor=darkblue, 
	urlcolor=darkblue, 
}

\newcommand{\parheading}[1]{{\bfseries #1.}}

\crefname{algocf}{alg.}{algs.}
\Crefname{algocf}{Algorithm}{Algorithms}

\crefname{appsec}{Appendix}{Appendices}

\let\originalleft\left \let\originalright\right
\renewcommand{\left}{\mathopen{}\mathclose\bgroup\originalleft}
\renewcommand{\right}{\aftergroup\egroup\originalright}

\renewcommand{\geq}{\geqslant}

\DeclareFontFamily{U}{mathx}{\hyphenchar\font45} \DeclareFontShape{U}{mathx}{m}{n}{<-> mathx10}{}
\DeclareSymbolFont{mathx}{U}{mathx}{m}{n} \DeclareMathAccent{\widebarc}{0}{mathx}{"73}

\title
      {BijectiveRemesh: Maintaining Bijective Mappings for Data Transfer Across
       Remeshed Manifolds}

\author{Leyi Zhu}
\affiliation{%
  \institution{New York University}
  \country{USA}
}

\author{Michael Tao}
\affiliation{%
  \institution{New York University}
  \country{USA}
}

\author{Yixin Hu}
\affiliation{%
  \institution{Tencent America}
  \country{USA}
}

\author{Daniele Panozzo}
\affiliation{%
  \institution{New York University}
  \country{USA}
}

\author{Denis Zorin}
\affiliation{%
  \institution{New York University}
  \country{USA}
}

\renewcommand{\shortauthors}{Zhu et al.}

\begin{abstract}
  We introduce \textit{BijectiveRemesh}, a robust algorithm for maintaining a
continuous, bijective mapping across complex remeshing sequences on both 2D triangle
surfaces and 3D tetrahedral meshes. Unlike traditional data transfer methods
that rely on interpolation or projection, our approach constructs a mathematically
rigorous composite map $f: \mathcal{M}_{\text{input}}\to \mathcal{M}_{\text{output}}$
by chaining local bijective atlases defined for each primitive operation. 

Our framework represents the overall mapping as a composition of local bijective atlases, one per remeshing operation. Building upon successive self-parameterization (SSP)~\cite{Liu2021}, we introduce a \textit{Shared Scaffold} structure for 2D triangle meshes that enforces global bijectivity through local orientation preservation. We extend this approach to handle edge splits, edge swaps, and vertex smoothing beyond the original edge collapses. For 3D tetrahedral meshes, we generalize the local atlas construction using Steinitz's Theorem and Maxwell-Cremona lifting to ensure valid embeddings. This enables exact tracking of geometric entities—points, curves, and surfaces—across remeshing, with applications from texture transfer to volumetric simulations.

\end{abstract}

\begin{document}

\begin{teaserfigure}
  \centering
  \includegraphics[width=0.95\linewidth]{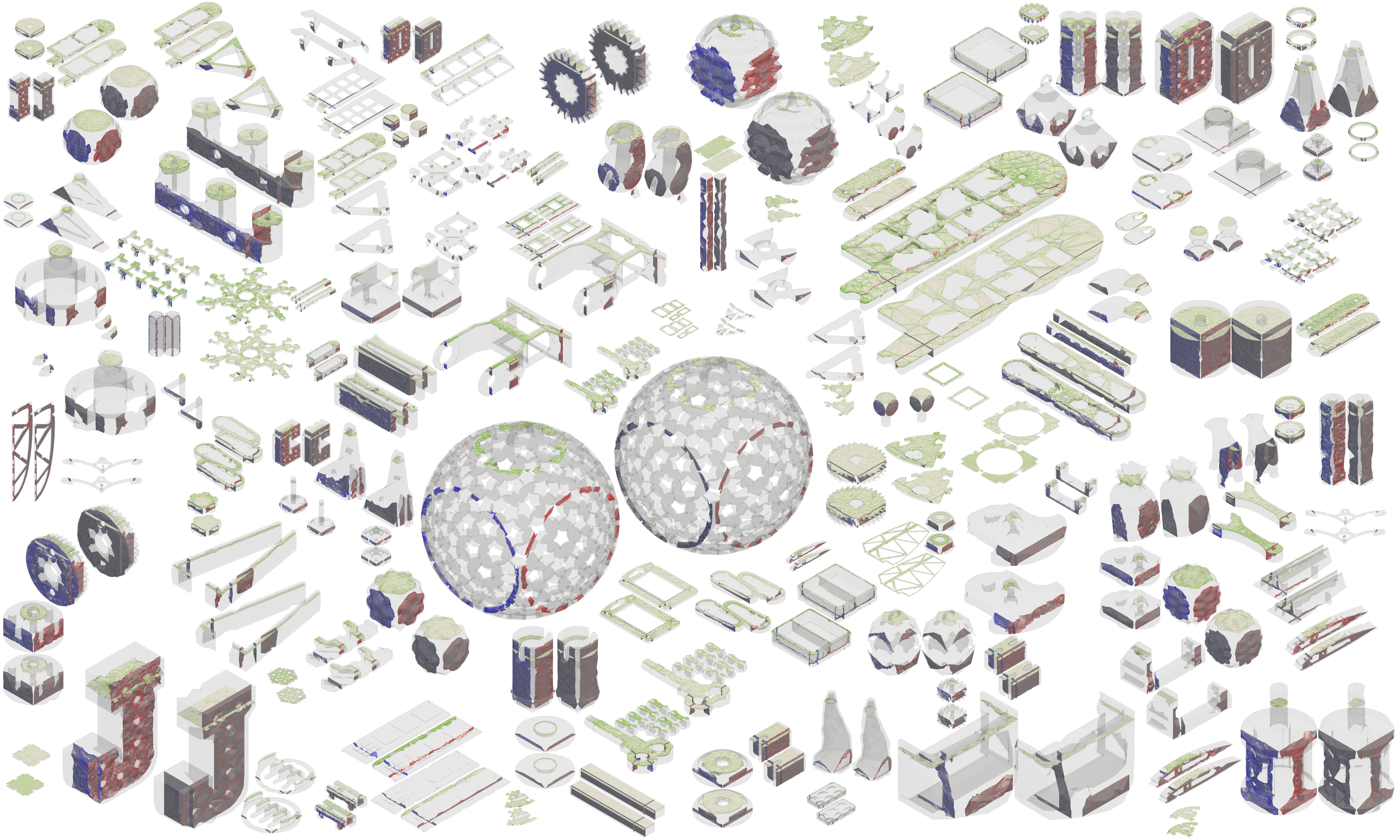}
  \caption{\textbf{Bijective surface tracking through tetrahedral mesh
  simplification.} We track axis-aligned planar surfaces (parallel to the $xy$, $x
  z$, and $yz$ planes) through tetrahedral mesh simplification on models from the
  Thingi10K dataset~\cite{Zhou2016}. For each model, surfaces are sampled on the
  simplified output mesh $\mathcal{M}_{\text{output}}$ (right) and back-tracked to
  the original mesh $\mathcal{M}_{\text{input}}$ (left). Our bijective framework
  rigorously preserves the intersection topology among surfaces: where surfaces intersect
  on the output, they intersect consistently on the input, maintaining the combinatorial
  structure of intersection curves throughout the tracking process.}
  \label{fig:teaser}
\end{teaserfigure}

\maketitle

\pagestyle{plain}
\thispagestyle{plain}


\section{Introduction}
Remeshing is ubiquitous in computer graphics and scientific computing, encompassing
tasks such as quality improvement, resolution adaptation, and artifact repair. A
fundamental challenge across all these operations is \textit{data transfer}:
propagating scalar fields, texture coordinates, material attributes, or boundary
conditions from the original mesh to the remeshed version with high fidelity. As
meshes undergo decimation, refinement, or optimization, maintaining precise correspondences
between initial and remeshed states is critical to prevent information loss and numerical
diffusion.

Conventional solutions for maintaining correspondences between meshes~\cite{Kobbelt1998}
rely on heuristics such as barycentric interpolation or closest-point
projections~\cite{Botsch2010}. While effective for minor adjustments, these methods
lack topological guarantees. Furthermore, when topological operations induce
substantial changes (e.g., edge collapses or swaps), these heuristic
correspondence often fail to be bijective, causing artifacts like UV seam
tearing, texture deviation, or loss of feature lines. Constructing a global
bijective parameterization could resolve these issues, but is computationally prohibitive
and prone to failure on complex non-disk-topology shapes~\cite{Kraevoy2004}.

We propose a scalable framework that maintains strict bijective mappings throughout
remeshing, enabling accurate data transfer between the original and remeshed
states. Rather than computing a single global parameterization, our method encodes the
overall map as a composition of local maps induced by atomic operation. This mesh
evolution "history", enables precise bidirectional data transfer regardless of
connectivity changes.

Our approach builds on two key ideas. First, for 2D triangle meshes, we construct
local atlases using a shared scaffold structure~\cite{Jiang2017} that ensures
bijectivity through locally injective optimization~\cite{Rabinovich2017}.
Second, for 3D tetrahedral meshes, we leverage Steinitz's Theorem~\cite{ribo2011small}
to construct convex polyhedral embeddings that geometrically prevent tetrahedral
overlaps. These bijective local atlases enable robust tracking of geometric entities—including
points, curves, and surfaces—throughout complex remeshing sequences. We
demonstrate our framework on applications including texture transfer across
remeshed models~\Cref{fig:texture_transfer_ogre,fig:texture_transfer_spot} and
preserving organ segmentations in medical CT data under mesh adaptation (\Cref{fig:ct_tracking_result}) and validate our topology preservation on the Thingi10K dataset~\citep{Zhou2016}.

\section{Related Work}
\label{sec:related}

Maintaining bijective correspondences is a long-standing challenge in geometry processing,
spanning surface parameterization, volumetric mapping, and attribute transfer.

\parheading{Strict Surface Homeomorphisms}
Bijective surface parameterization traditionally relied on Tutte’s embedding theorem~\cite{Tutte1963},
which guarantees an injective mapping for 3-connected planar graphs into convex
domains using barycentric coordinates~\cite{Floater2003}. While robust, these
linear methods are restricted to fixed convex boundaries. To allow boundary movement
and minimize isometric distortion, non-linear optimization frameworks such as
SLIM~\cite{Rabinovich2017} and Total Lifted Content (TLC)~\cite{du2020lifting} have
been proposed. However, these methods either require an already-injective
initialization or lack global bijectivity guarantees on complex topologies. Our 2D
Shared Scaffold structure builds upon the Simplicial Complex Augmentation Framework
(SCAF)~\cite{Jiang2017}, which ensures global bijectivity by reducing the
problem to local orientation preservation within an augmented tessellation of
the ambient space~\cite{Lipman2014}.

\parheading{Volumetric Mapping and Topological Obstructions}
Extending bijectivity guarantees to 3D tetrahedral meshes is significantly more complex
due to topological obstructions. Notably, the 3D analog of Tutte’s theorem does not
hold universally; even with convex boundaries, because internal tetrahedra can invert if
the mesh contains specific forbidden minors like $K_{6}$ or $K_{3,3,1}$~\cite{Alexa2023,
Floater2006}. Constructive methods such as Simplicial Foliations~\cite{Campen2016}
and the Shrink-and-Expand (SaE) framework \cite{Nigolian2023, Nigolian2024}
provide theoretical guarantees for shellable meshes but often incur extreme
computational costs and memory-intensive mesh refinement. In contrast, our
approach utilizes Steinitz’s Theorem~\cite{Steinitz1922} and Maxwell-Cremona
lifting~\cite{ribo2011small} to construct local convex polyhedral embeddings for
boundary operations. This ensures validity by exploiting the geometric property that
convexity of the boundary polyhedron prevents interior overlap, avoiding the
need for expensive global constructive schemes.

\parheading{Mapping Preservation and Attribute Transfer}
The "Bijective Prism Shell"~\cite{Jiang2020, Jiang2021} establishes a common domain
for spatially close surfaces by constructing a volumetric shell between them,
though the approach is limited to surfaces within a narrow distance threshold
and requires careful shell thickness tuning. Another direction is the use of "common
triangulations"~\cite{Schmidt2023}, which decouple map resolution from input
complexity by adaptively refining a shared triangulation. However, this approach
is restricted to genus-0 surfaces (mapping via the sphere) and loses the
remeshing tracking sequence, requiring manual selection of landmark point pairs. Alternative approximate methods
such as reversible harmonic maps~\cite{Ezuz2019} and low-resolution
correspondence~\cite{maggioli2024rematching} trade accuracy for efficiency, but inherently
incur information loss during transfer due to their continuous relaxation or downsampling
strategies.

The closest work to ours is the successive self-parameterization (SSP) framework~\cite{Liu2021},
which constructs bijective mappings by composing local atlases across remeshing
operations. The atlases are constructed by a joint-flattening strategy where the
3D coordinates of pre- and post-operation patches (the localized mesh regions affected
by each operation) are simultaneously parameterized into a shared 2D domain. While
this approach successfully distributes distortion between mesh states, it
suffers from two key limitations: First, the optimization lacks robustness
guarantees and can fail to converge or produce inverted elements when mesh
quality is poor. We observe frequent failures on models from the Thingi10K dataset~\cite{Zhou2016}—for
instance, model \#1706476 exhibits failure of SSP because it prevents any collapse
operations from completing (\Cref{fig:ssp_comparison}). Second, SSP is designed
specifically for edge collapse operations in the context of mesh coarsening and
does not generalize to other remeshing primitives that we support, such as edge splits
(for refinement), edge swaps (for quality improvement), or vertex smoothing (for
geometry optimization). Additionally, the framework does not extend to 3D tetrahedral
meshes, where boundary operations require fundamentally different geometric
constructions.

\begin{figure}[htb]
    \centering
    \begin{subfigure}
        [b]{0.48\linewidth}
        \centering
        \includegraphics[width=\linewidth]{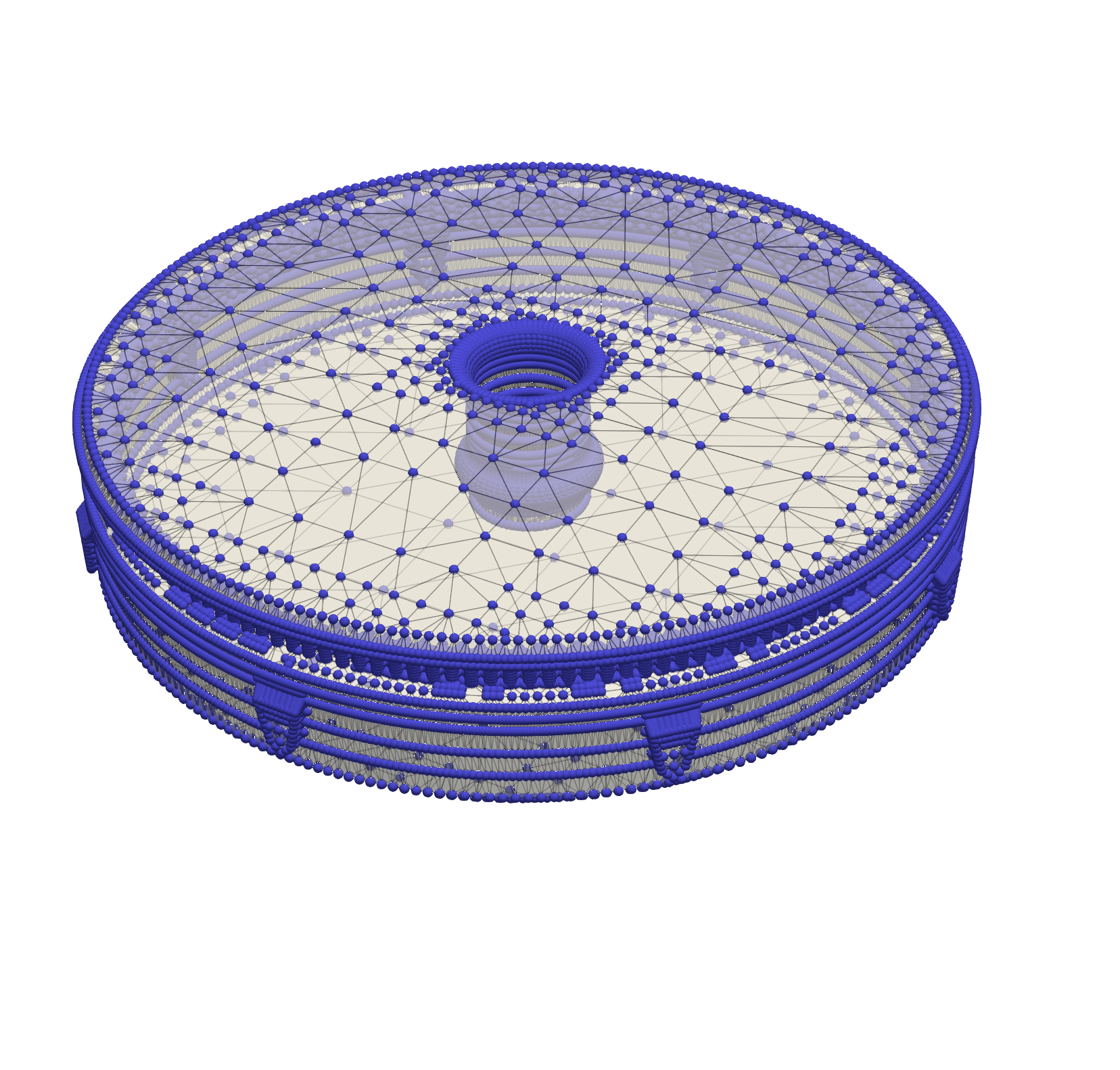}
        \caption{SSP method fails to complete any edge collapses}
        \label{fig:ssp_fail}
    \end{subfigure}
    \hfill
    \begin{subfigure}
        [b]{0.48\linewidth}
        \centering
        \includegraphics[width=\linewidth]{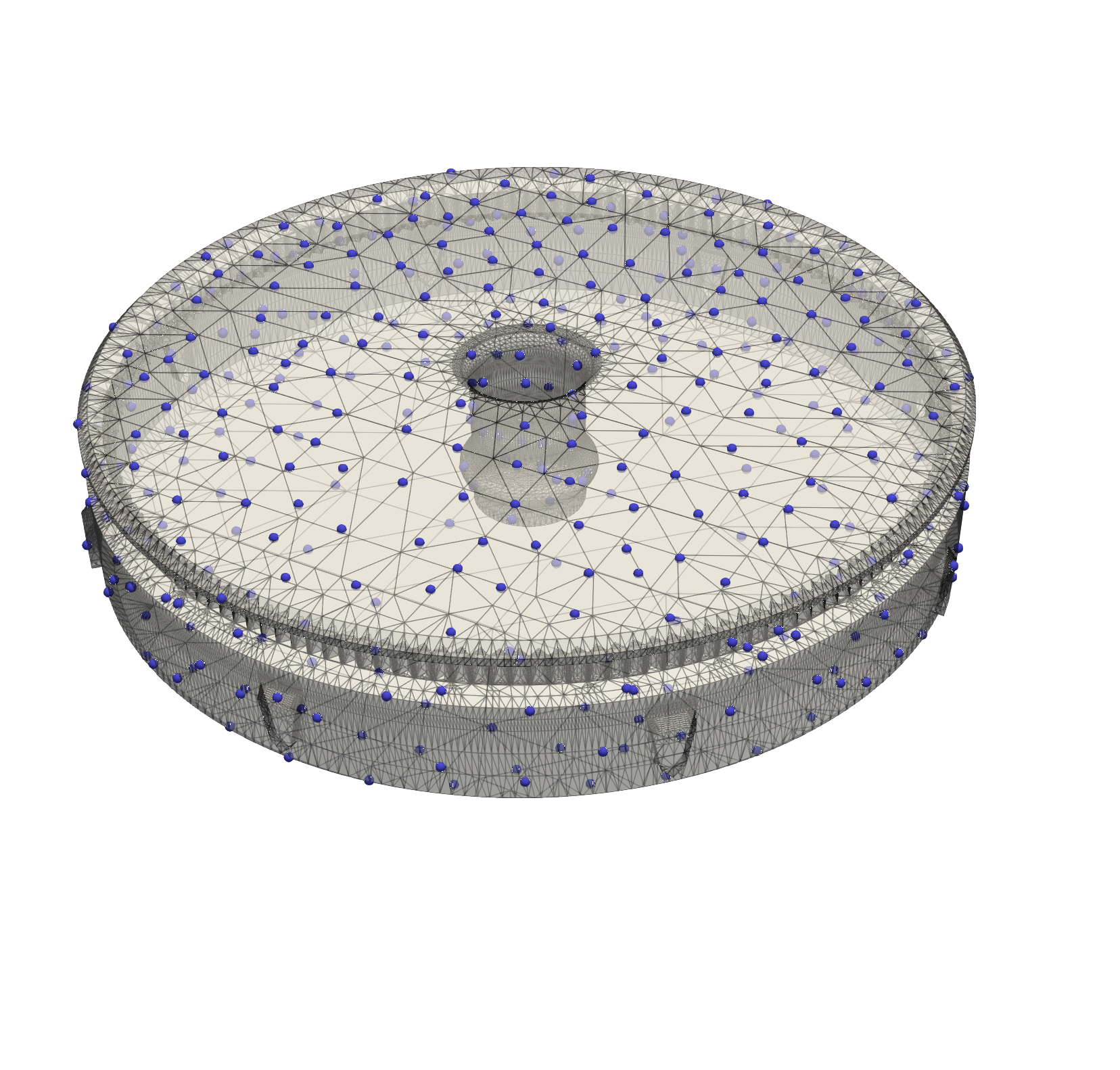}
        \caption{Our method robustly generates coarse-to-fine mapping}
        \label{fig:our_success}
    \end{subfigure}
    \caption{Comparison on Thingi10K model \#1706476. We visualize the coarse-to-fine
    mapping by transferring points from the decimated mesh back to the original fine
    mesh using successive parameterization. \textbf{Left:} The SSP framework~\cite{Liu2021}
    fails during optimization, preventing any edge collapse operations from
    completing. \textbf{Right:} Our shared scaffold method successfully
    constructs bijective mappings throughout the decimation process, enabling robust
    data transfer between mesh representations.}
    \label{fig:ssp_comparison}
\end{figure}

\section{Method}
\label{sec:method}

Given an input manifold mesh $\mathcal{M}_{\text{input}}$ and a sequence of remeshing
operations $\mathcal{O}= \{o_{1}, o_{2}, \dots, o_{n}\}$, our goal is to
maintain a continuous, bijective mapping $f: \mathcal{M}_{\text{input}}\to \mathcal{M}
_{\text{output}}$ throughout the entire remeshing process. We assume the input
$\mathcal{M}$ is a manifold simplicial complex, which can be either a 2D triangle
mesh embedded in $\mathbb{R}^{3}$ or a 3D tetrahedral mesh in $\mathbb{R}^{3}$.
The operation sequence $\mathcal{O}$ may consist of various topological and
geometric modifications, including edge collapses, edge splits, edge swaps, and
vertex smoothing.

We follow the insight of \citet{Liu2021} by noting that, while the cumulative
effect of remeshing is global, each individual operation $o_{i}: \mathcal{M}_{i-1}
\to \mathcal{M}_{i}$ only acts on a small, localized patch (the Region of Interest,
Figure~\ref{fig:2d_ops}). We denote these local patches as
$\mathcal{P}_{i}^{\text{before}}\subset \mathcal{M}_{i-1}$ and
$\mathcal{P}_{i}^{\text{after}}\subset \mathcal{M}_{i}$, corresponding to the mesh
state before and after the operation.

To construct a bijection between these two patches, we embed both into a shared
geometric domain $\mathcal{C}_{i}$ via embeddings
$e_{i}^{\text{before}}: \mathcal{P}_{i}^{\text{before}}\to \mathcal{C}_{i}$ and
$e_{i}^{\text{after}}: \mathcal{P}_{i}^{\text{after}}\to \mathcal{C}_{i}$.~(Figure~\ref{fig:curve_tracking})
This establishes a canonical bijection:
\begin{equation}
  \varphi_{i}= (e_{i}^{\text{after}})^{-1}\circ e_{i}^{\text{before}}: \mathcal{P}_{i}^{\text{before}}
  \to \mathcal{P}_{i}^{\text{after}}.
\end{equation}
Since the operation only modifies the local region, this local bijection extends
to a global map $\varphi_{i}: \mathcal{M}_{i-1}\to \mathcal{M}_{i}$ that is the identity
outside of $\mathcal{P}_{i}^{\text{before}}$. The global bijective mapping is
then obtained by composition:
\begin{equation}
  \label{eq:global-composition}\Phi = \varphi_{n}\circ \varphi_{n-1}\circ \dots \circ \varphi_{1}.
\end{equation}

\noindent
This formulation allows us to handle diverse operations---including splits for
refinement and collapses for coarsening---within a unified framework.

\subsection{Constructing Bijective Maps Using Local Atlases}
\label{subsec:local-atlas-construction} We use different bijective maps for
triangle meshes and tetrahedral mesh, but they follow the same fundamental principle
of constructing a shared parametric space through which we can perform mapping.
The main difficulty is that remeshing can change the domain of the mesh, so our bijective
maps must be able to handle a certain amount of distortion. For triangle meshes
we are minimizing the distortion of points in 3D and most operations introduce
some distortion to the geometry that must be minimized. On the other hand,
tetrahedral meshes are embedded in 3D so for operations on the interior the
identity map to $\mathbb{R}^{3}$ is the perfect bijective map and the entire
challenge is how to deal with operations on the boundary that distort the domain
of the tetrahedral mesh.

\subsubsection{2D: Triangle Mesh}
\label{subsubsec:2d-triangle-mesh}
\parheading{Background: Joint Flattening}

Our 2D approach builds upon the joint flattening strategy introduced by~\cite{Liu2021}.
The key idea is to leverage the fact that when the patch $\mathcal{P}^{\text{before}}$
of a local operation lies on the interior of $\mathcal{M}_{i-1}$ then $\mathcal{P}
^{\text{after}}$ lies on the interior of $\mathcal{M}_{i}$ and their boundaries remain
geometrically consistent:
$\partial \mathcal{P}^{\text{before}}\equiv \partial \mathcal{P}^{\text{after}}$.
This allows both patches to be parameterized into a shared 2D domain by minimizing
a joint distortion energy:
\begin{equation}
  \label{eq:joint-distortion}\min_{\mathbf{U}}\quad E(\mathcal{P}^{\text{before}}
  , \mathbf{U}^{\text{before}}) + E(\mathcal{P}^{\text{after}}, \mathbf{U}^{\text{after}}
  )
\end{equation}
subject to shared boundary coordinates:
\begin{equation}
  \label{eq:boundary-constraint}\mathbf{u}_{v}^{\text{before}}= \mathbf{u}_{v}^{\text{after}}
  , \quad \forall v \in \partial \mathcal{P}.
\end{equation}
where $\mathbf{U}^{\text{before}}$ and $\mathbf{U}^{\text{after}}$ are the UV coordinates
for all vertices in $\mathcal{P}^{\text{before}}$ and $\mathcal{P}^{\text{after}}$
respectively, $\mathbf{u}_{v}\in \mathbb{R}^{2}$ denotes the UV position of
vertex $v$, and $E$ is a standard distortion metric such as Symmetric Dirichlet~\cite{Smith2015}
or ARAP energy~\cite{Sorkine2007}.

For edge collapses on the boundary, where vertex $i$ is collapsed toward vertex
$j$, a colinearity constraint is applied: the collapsed position
$\mathbf{u}_{i}$ must remain colinear with its neighboring boundary vertices $j$
and $k$ in the parametric domain. Although this treatment ensures the boundary shape
remains well-defined for both connectivity states, it does not guarantee
bijectivity: the optimization can produce overlapping or inverted triangles,
particularly on poor-quality meshes. To rigorously enforce bijectivity, we
augment this framework with a \textit{shared scaffold structure}. We first
review the Simplicial Complex Augmentation Framework (SCAF)~\citep{Jiang2017}
that provides the theoretical foundation, then describe our construction of the
shared scaffold for joint flattening.

\parheading{Background: Simplicial Complex Augmentation Framework}
SCAF guarantees bijectivity by adding an auxiliary "scaffold" triangulation around
a mesh patch to make it so that global bijectivity can preserved by maintaining local
injectivity. Given a 2D mesh patch $\mathcal{P}$ to be parameterized, SCAF constructs
a scaffold $\mathcal{S}$ between $\mathcal{P}$ and its bounding box so $\mathcal{D}
= \mathcal{S}\cup \mathcal{P}$ is a simplicial complex that tessellates a convex
domain (\Cref{fig:shared-scaffold-opt}).

The critical property of SCAF is that any potential self-intersections of $\mathcal{P}$
are captured by checking for \emph{local injectivity} (i.e., has positive Jacobian
determinant per element) on $\mathcal{D}$ so local and global bijectivity become
identical. then the entire mapping is \emph{globally bijective}. As such, global
overlaps can be reduced to per-element orientation checks. We can therefore apply
locally injective optimization methods like SLIM~\cite{Rabinovich2017} with
an orientation-preserving line search to optimize distortion while guaranteeing global
bijectivity.

\parheading{Shared Scaffold for Joint Flattening}
Since $\mathcal{P}^{\text{before}}$ and $\mathcal{P}^{\text{after}}$ are both
parameterized into the same 2D domain and share a common boundary
$\partial\mathcal{P}$, we can naturally construct a \emph{single} shared
scaffold $\mathcal{S}$ for \emph{both} patches. Specifically, we form two augmented
complexes:
\begin{equation}
  \mathcal{D}^{\text{before}}= \mathcal{P}^{\text{before}}\cup \mathcal{S}, \qquad
  \mathcal{D}^{\text{after}}= \mathcal{P}^{\text{after}}\cup \mathcal{S}.
\end{equation}

We then minimize the joint distortion energy over the vertex positions of both augmented
complexes:
\begin{equation}
  \min_{\mathbf{U}}\quad E(\mathcal{D}^{\text{before}}, \mathbf{U}^{\text{before}}
  ) + E(\mathcal{D}^{\text{after}}, \mathbf{U}^{\text{after}}),
\end{equation}
subject to:
\begin{itemize}
  \item \textbf{Shared boundary}: $\mathbf{u}_{v}^{\text{before}}= \mathbf{u}_{v}
    ^{\text{after}}$ for all $v \in \partial\mathcal{P}$.

  \item \textbf{Shared scaffold}: $\mathbf{u}_{v}^{\text{before}}= \mathbf{u}_{v}
    ^{\text{after}}$ for all $v \in \mathcal{S}$.
\end{itemize}

By using SCAF and SLIM the result of our optimization is guaranteed to produce a
globally bijective embedding of $\mathcal{P}^{\text{before}}$ and
$\mathcal{P}^{\text{after}}$. Furthermore, since both patches share the same non-overlapping
scaffold $\mathcal{S}$, a bijective correspondence between $\mathcal{P}^{\text{before}}$
and $\mathcal{P}^{\text{after}}$ can be established through the resulting common
parametric domain. Figure~\ref{fig:shared-scaffold-opt} illustrates this shared
scaffold framework.

\parheading{Application to Different Operation Types}
The shared scaffold framework applies uniformly across various remeshing operations.
Figure~\ref{fig:2d_ops} illustrates the local patches for each operation
type. Using the open star $\text{St}$ (\citep{Munkres_2018}), the local patches we
use are:

\begin{enumerate}
  \item \textit{Edge Collapse} $(i,j) \to k$: $\mathcal{P}^{\text{before}}= \text{St}
    (i) \cup \text{St}(j)$ (union of 1-rings of both vertices),
    $\mathcal{P}^{\text{after}}= \text{St}(k)$ (1-ring of merged vertex).

  \item \textit{Edge Split} on edge $(i,j)$ creating vertex $k$: $\mathcal{P}^{\text{before}}
    = \text{St}(i,j)$ (triangles incident to the edge),
    $\mathcal{P}^{\text{after}}= \text{St}(k)$ (1-ring of new vertex).

  \item \textit{Edge Swap}: $\mathcal{P}^{\text{before}}$ and $\mathcal{P}^{\text{after}}$
    are both the two triangles sharing the edge.

  \item \textit{Vertex Smoothing} of vertex $v$: $\mathcal{P}^{\text{before}}= \mathcal{P}
    ^{\text{after}}= \text{St}(v)$ (1-ring of the vertex, connectivity unchanged).
\end{enumerate}

\begin{figure}[htb]
  \centering
  \begin{subfigure}
    [b]{0.9\linewidth}
    \centering
    \includegraphics[width=\linewidth]{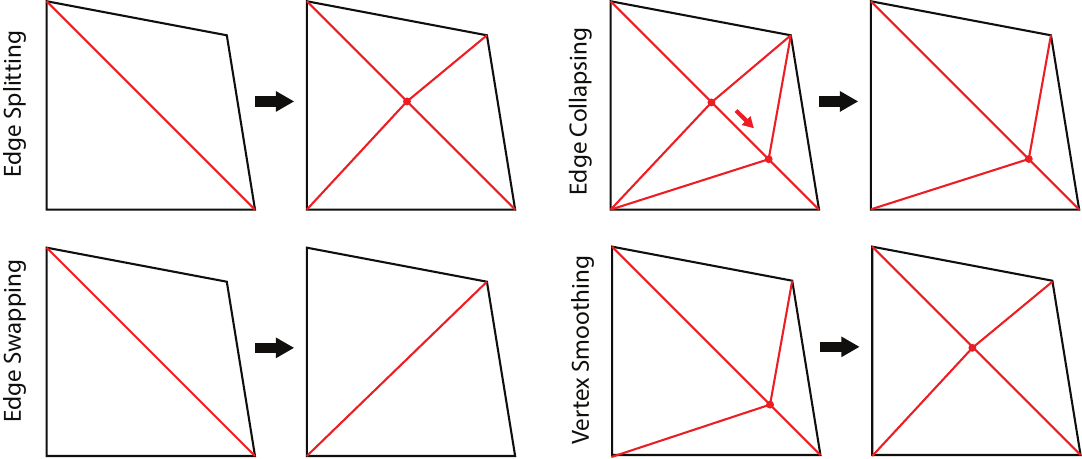}
  \end{subfigure}
  \caption{\textbf{Patches for different remeshing operations.}}
  \label{fig:2d_ops}
\end{figure}


\begin{figure}[t]
  \centering
  \begin{subfigure}
    [b]{0.58\linewidth}
    \centering
    \includegraphics[width=\linewidth]{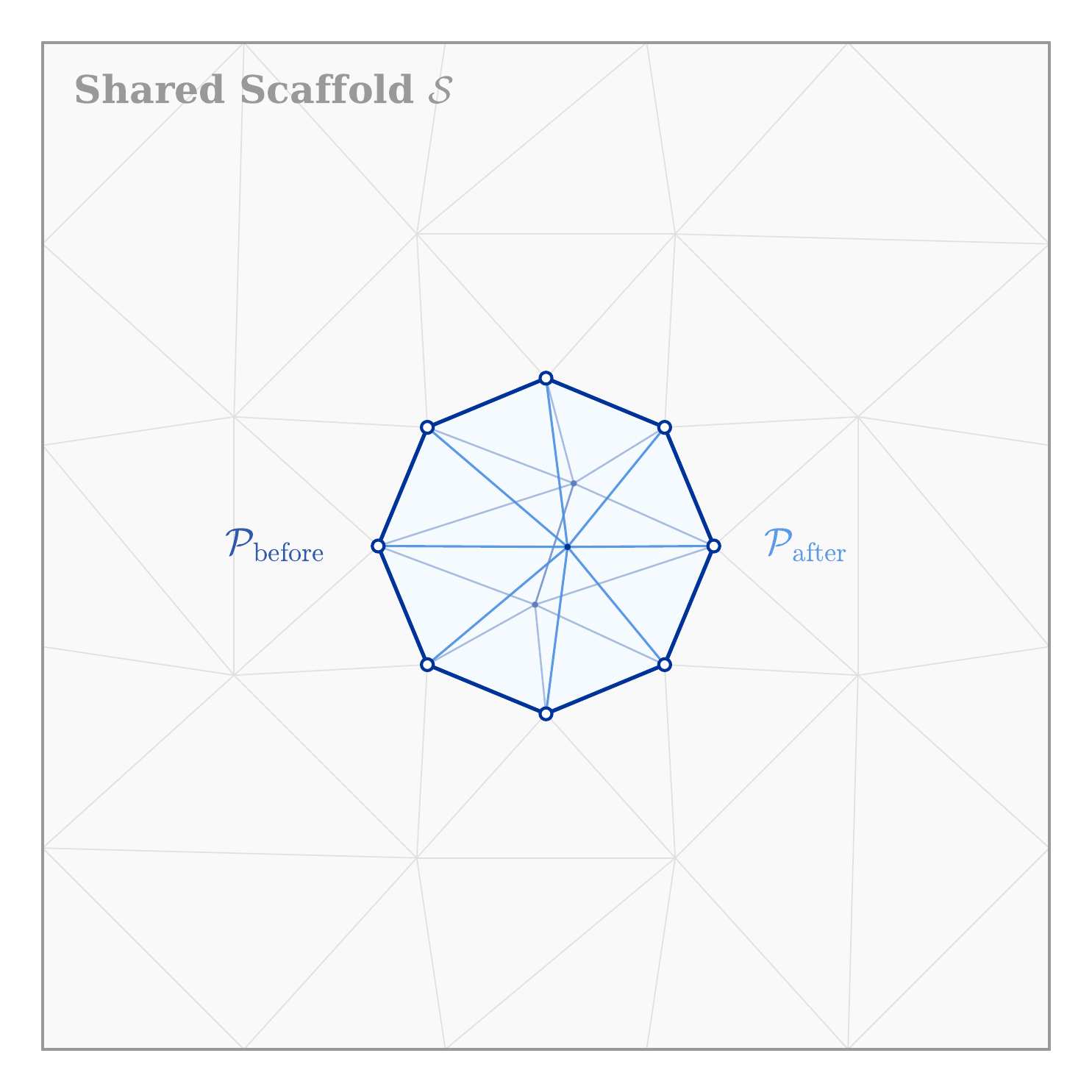}
    \caption{Shared Scaffold Structure}
    \label{fig:scaf-init}
  \end{subfigure}
  \hfill
  \begin{subfigure}
    [b]{0.4\linewidth}
    \centering
    \includegraphics[width=\linewidth]{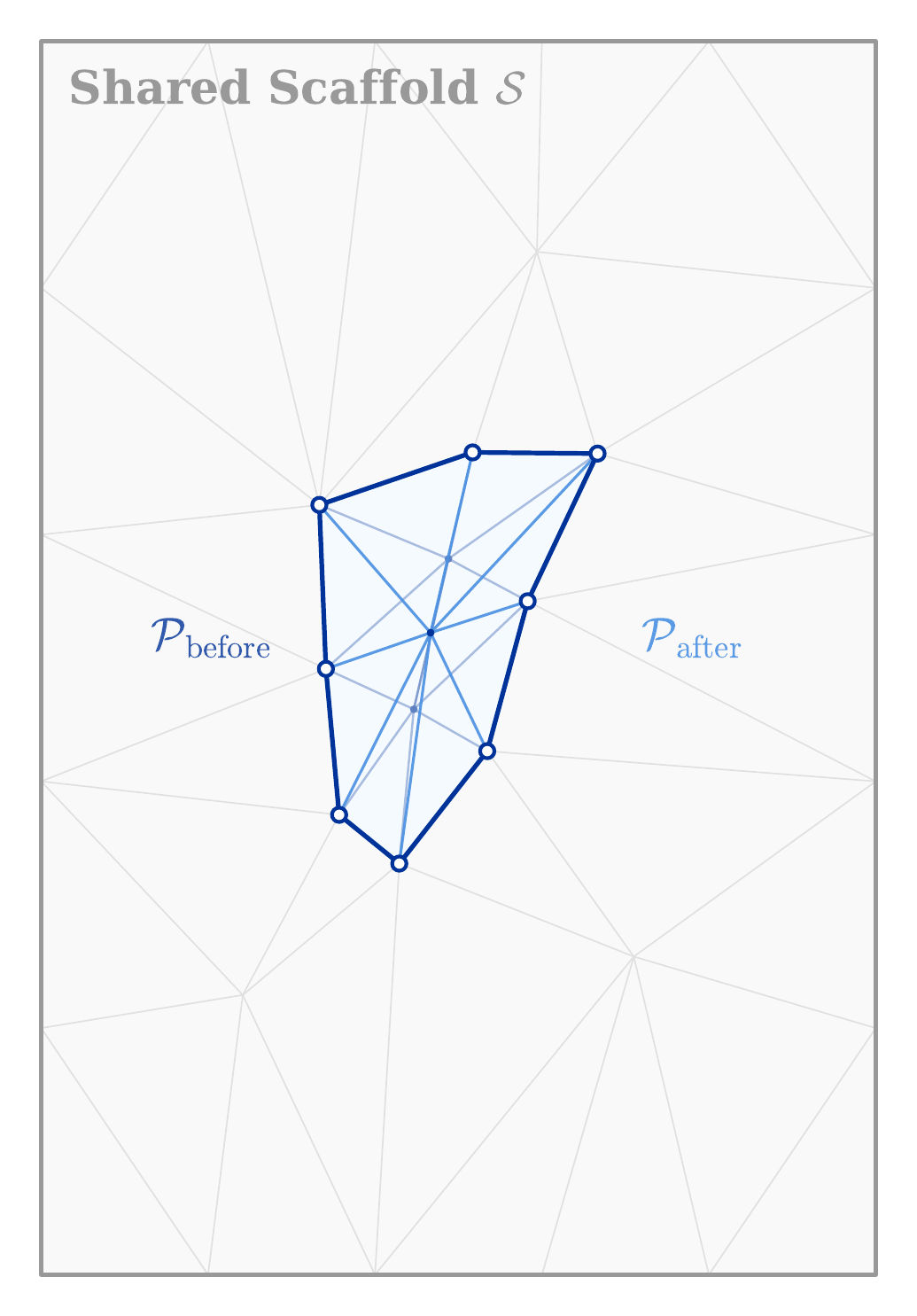}
    \caption{After Joint Optimization}
    \label{fig:scaf-opt}
  \end{subfigure}

  \caption{\textbf{Shared scaffold framework for bijective atlas construction.} (a)
  Both $\mathcal{P}^{\text{before}}$ and $\mathcal{P}^{\text{after}}$ are embedded
  within the same scaffold $\mathcal{S}$ (gray), sharing a common boundary
  $\partial\mathcal{P}$ (dark blue). (b) After joint optimization, the scaffold
  $\mathcal{S}$ is deformed to minimize distortion while maintaining the sharing
  constraint ($\mathbf{u}_{v}^{\text{before}}= \mathbf{u}_{v}^{\text{after}}$ for
  all $v \in \mathcal{S}\cup \partial\mathcal{P}$). Local injectivity via SLIM optimization
  guarantees global bijectivity for both augmented complexes through the SCAF property,
  establishing bijective correspondence via the shared parametric domain.}
  \label{fig:shared-scaffold-opt}
\end{figure}

\subsubsection{3D: Tetrahedral Mesh}
\label{subsubsec:3d-tetrahedral-mesh}

Recall that for tetrahedral meshes interior operations constructing
bijective mapping is trivial and no deformation is necessary at all. In such
cases, we extract the affected region (the closed star of the operation) as the
local atlas and use identity mapping between the old and new connectivity to
maintain bijectivity.

The primary challenge arises when handling \textbf{boundary operations}, which
introduce geometric changes to the mesh surface. Of the four operations we choose use in our system, three can be implemented concisely:

\parheading{Boundary Edge Split}
Splitting a boundary edge does not alter the geometric shape of the boundary surface and 
only refine the mesh connectivity. Since the boundary geometry remains
unchanged we can extract the affected local region and establish a canonical bijective
correspondence, identical to the treatment of interior operations.

\parheading{Boundary Vertex Smoothing}
Vertex smoothing preserves the mesh connectivity while modifying vertex
positions. In this case the set of tetrahedra in $\mathcal{P}^{\text{before}}$
and $\mathcal{P}^{\text{after}}$ have identical combinatorial structure and we
presume that when a vertex of a tetrahedron is moved the points in the
tetrahedron are all moved linearly. As such, we move each point in each
tetrahedron so that its coordinate with respect to that tetrahedron's
barycentric coordinates remains the same before and after the operation. 

\parheading{Boundary Edge Swap}
An edge swap decomposes into edge split followed by edge collapse. As the split introduces no distortion and we have to implement the collapse anyway we chose to implement the swap atlas as the composition of the split atlas and the collapse atlas.

\subsubsection{Boundary Edge Collapse}
The remaining atlas to describe is the \textbf{boundary edge collapse}, which
genuinely alters the boundary geometry and connectivity simultaneously.
The object is, therefore, to construct a reparametrization that embeds both
$\mathcal{P}^{\text{before}}$ and $\mathcal{P}^{\text{after}}$ into a consistent
geometric domain, thereby establishing a bijective correspondence between them.

\label{subsubsec:boundary-edge-collapse} Consider a mesh $\mathcal{M}_{\ell}$ at step
$\ell$ of the remeshing sequence, with boundary surface $\partial \mathcal{M}_{\ell}$.
Suppose we perform a boundary edge collapse on edge $(i,j)$, yielding the
updated mesh $\mathcal{M}_{\ell+1}$. Without loss of generality, we assume vertex
$i$ is collapsed toward vertex $j$.


\parheading{Local Patch Extraction}
Consider a boundary edge collapse operation that collapses vertex $i$ toward vertex $j$.
Let $\text{St}_{\ell}(v)$ denote the star of vertex $v$ in $\mathcal{M}_{\ell}$.
We define the pre-collapse local patch as
\begin{equation}
\mathcal{P}^{\text{before}} = \text{St}_{\ell}(i).
\end{equation}
The edge collapse removes every tetrahedron incident to edge $(i,j)$ and replaces vertex $i$ by $j$ in the remaining tetrahedra. The post-collapse local patch $\mathcal{P}^{\text{after}} \subset \mathcal{M}_{\ell+1}$ consists of exactly those surviving tetrahedra, i.e., the tetrahedra in $\text{St}_{\ell}(i)$ that are not incident to edge $(i,j)$, with vertex $i$ replaced by $j$.

\parheading{Coplanarity Constraint}
Similar to the 2D case for handling boundary edges, where we impose colinearity
constraints to ensure a consistent boundary curve, in 3D we adopt an analogous 
strategy: during reparametrization, we constrain the boundary triangles affected 
by the collapse to lie in a common plane. Specifically, we require that the 
portions of the local patch boundary lying on the global mesh boundary,
\begin{equation}
  \partial \mathcal{P}^{\text{before}}\cap \partial \mathcal{M}_{\ell}\quad \text{and}
  \quad \partial \mathcal{P}^{\text{after}}\cap \partial \mathcal{M}_{\ell+1},
\end{equation}
are reparametrized to lie in the same plane. Therefore, 
the collapse operation will not change the shape of the parametrization domain.
Our problem then reduces to finding a reparametrization that satisfies this coplanarity 
constraint and can accommodate both the pre-collapse connectivity $\mathcal{P}^{\text{before}}$ 
and the post-collapse connectivity $\mathcal{P}^{\text{after}}$.

\parheading{Reduction to Convex Polyhedron Construction}
As established above, our problem reduces to finding a common reparametrization 
domain that can accommodate both $\mathcal{P}^{\text{before}}$ and 
$\mathcal{P}^{\text{after}}$. By Lemma~\ref{lem:non-overlapping}, this problem 
further simplifies to constructing a \textbf{convex polyhedron} that can 
simultaneously embed the boundaries $\partial \mathcal{P}^{\text{before}}$ and 
$\partial \mathcal{P}^{\text{after}}$. Crucially, to satisfy the coplanarity 
constraint, we require that the portions
$\partial \mathcal{P}^{\text{before}}\cap \partial \mathcal{M}_{\ell}$ and $\partial \mathcal{P}^{\text{after}}\cap \partial \mathcal{M}_{\ell+1}$
are embedded onto a single face of this polyhedron. Specifically, the one-ring 
boundary vertices of vertex $i$, together with vertex $i$ itself, must all lie on 
the same polyhedral face on the constructed convex polyhedron.

\parheading{Combinatorial Formulation} We formulate this as a combinatorial problem. The boundary surface $\partial \mathcal{P}^{\text{before}}$ (equivalently, $\partial \mathcal{P}^{\text{after}}$) consists of two parts: (1) the $k$ boundary triangles that do not contain vertex $i$, denoted $f_0, \ldots, f_{k-1}$, and (2) the one-ring of vertex $i$ on $\partial\mathcal{M}_{\ell}$, which forms an additional combinatorial face $f_k$ (e.g., the face $[0,6,5,4,7,1]$ in~\cref{fig:tutte-embedding}). Our question becomes: \emph{Can the combinatorial graph $G$ formed by $\{f_0, \ldots, f_k\}$ be realized as the 1-skeleton of a convex polyhedron?}

The following classical result provides an affirmative answer:

\begin{theorem}[Steinitz's Theorem]
  \label{thm:steinitz} A graph $G$ is the 1-skeleton of a convex polyhedron if and
  only if $G$ is simple, planar, and 3-connected.
\end{theorem}

\parheading{Constructive Algorithm}
We follow the approach of~\cite{ribo2011small} to realize a convex polyhedra embedding of the graph $\mathcal{G}$. We first select a triangle from $\{f_0, \ldots, f_{k-1}\}$ as the outer boundary and compute a planar Tutte embedding~\citep{Tutte1963} (\cref{fig:tutte-embedding}). The Maxwell-Cremona lifting~(\cref{fig:maxwell-cremona-lifting}) then elevates this 2D embedding into a 3D convex polyhedron by assigning heights to each vertex. Note that a na\"ive approach---e.g., placing boundary vertices on one side and lifting interior vertices uniformly---does not in general produce a convex polyhedron, because the induced boundary edge stresses may violate the sign conditions required by the Maxwell-Cremona correspondence~\cite{ribo2011small}. The full Tutte embedding with careful outer-face selection ensures these conditions are satisfied. The detailed pseudocodes are provided in Algorithms~\ref{alg:tutte-embedding} and~\ref{alg:maxwell-cremona} in the Appendix.

\begin{figure}[t]
  \centering
  \begin{subfigure}
    [b]{0.43\linewidth}
    \centering
    \includegraphics[width=\linewidth]{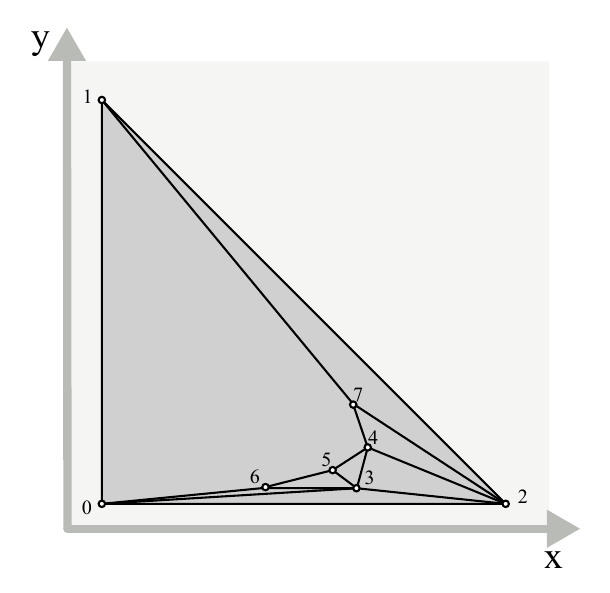}
    \caption{Tutte barycentric embedding}
    \label{fig:tutte-embedding}
  \end{subfigure}
  \hfill
  \begin{subfigure}
    [b]{0.53\linewidth}
    \centering
    \includegraphics[width=1\linewidth]{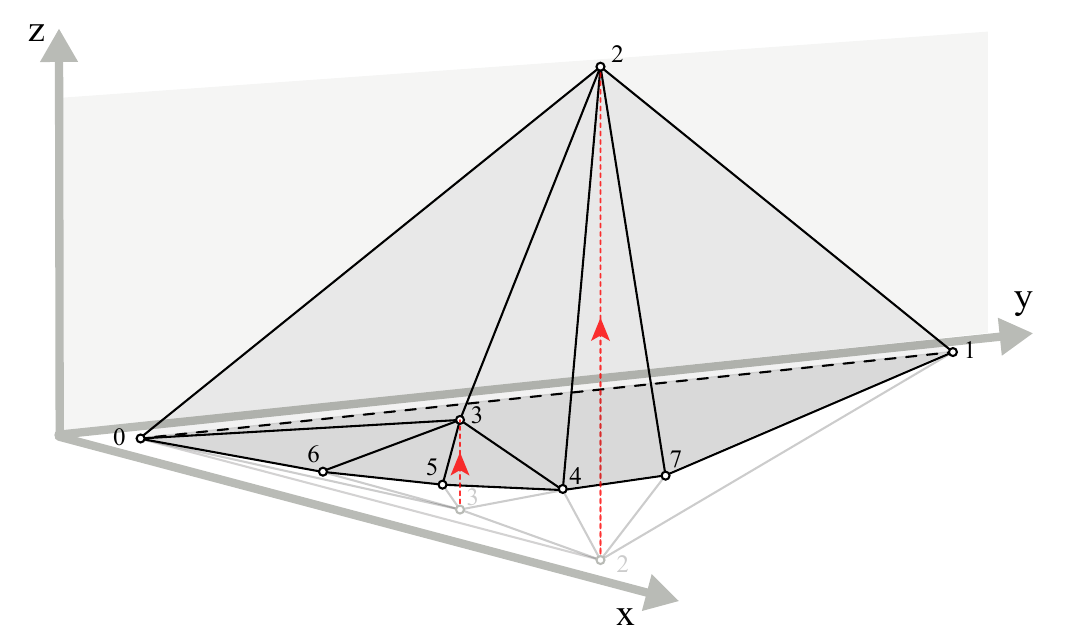}
    \caption{Maxwell-Cremona lifting}
    \label{fig:maxwell-cremona-lifting}
  \end{subfigure}
  \caption{Constructive algorithm for convex polyhedra embedding. (a) The Tutte embedding
  (Algorithm~\ref{alg:tutte-embedding}) produces a crossing-free 2D planar embedding
  with the outer triangle $f_{0}= [0, 1, 2]$ fixed at positions $(0,0)$, $(1,0)$,
  and $(0,1)$. Interior vertices (labeled 3--7) are positioned at their barycentric
  coordinates, yielding a valid planar layout. (b) The Maxwell-Cremona lifting (Algorithm~\ref{alg:maxwell-cremona})
  elevates this 2D embedding into a 3D convex polyhedron by assigning heights to
  each vertex.}
  \label{fig:convex-embedding-construction}
\end{figure}

\begin{figure}[t]
  \centering
  \includegraphics[width=0.98\linewidth]{
    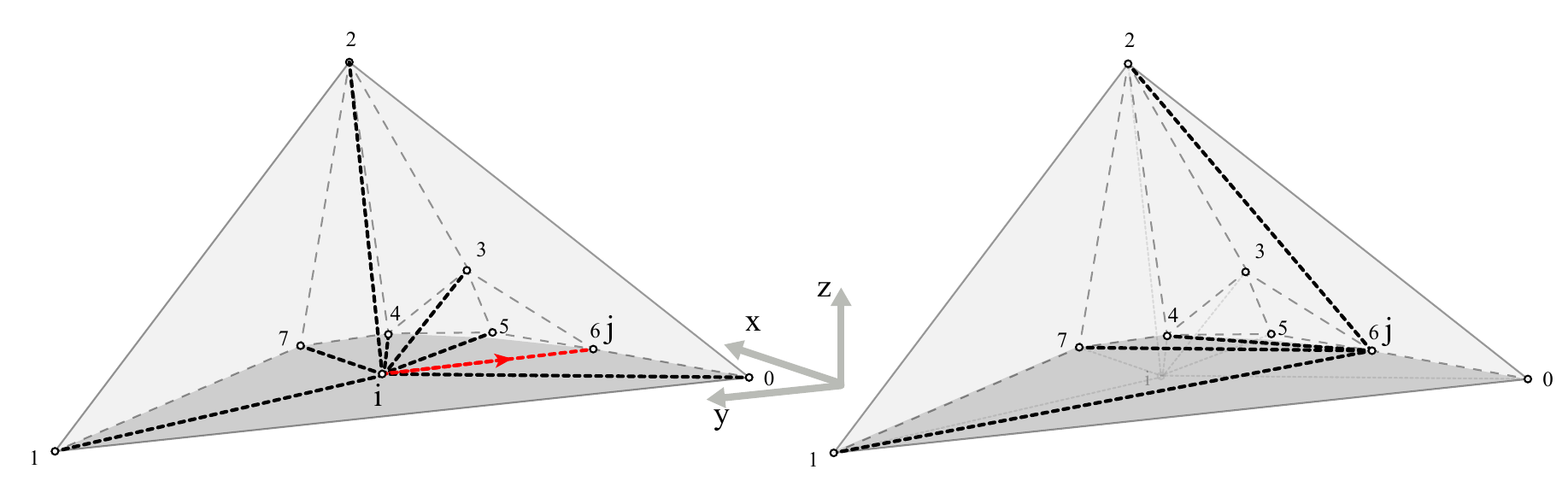
  }
  \caption{Initial valid embedding for both $\mathcal{P}^{\text{before}}$ and
  $\mathcal{P}^{\text{after}}$. After constructing the convex polyhedron $\mathcal{C}$
  via Algorithms~\ref{alg:tutte-embedding} and~\ref{alg:maxwell-cremona}, vertex
  $i$ is placed at the barycenter of face $f_{k}$ (the one-ring of $i$ on the boundary).
  By Lemma~\ref{lem:non-overlapping}, this configuration is valid (non-overlapping)
  for both the pre-collapse connectivity (with edge $(i,j)$) and the post-collapse
  connectivity (where $i$ has collapsed to $j$). This provides a valid starting point
  for joint optimization.}
  \label{fig:initial-embedding}
\end{figure}

\parheading{Initial Embedding and Joint Optimization}
A crucial requirement for our bijective optimization is that the initial embedding must be \emph{flip-free} (all elements have positive orientation). Our optimization employs an orientation-preserving line search that can \emph{maintain} this invariant but cannot \emph{recover} from inverted elements---the AMIPS energy diverges to $+\infty$ for degenerate or inverted tetrahedra. The convex polyhedron construction provides this essential flip-free initialization: having constructed the embedding of $\mathcal{G}$, we recover valid embeddings for both $\mathcal{P}^{\text{before}}$ and $\mathcal{P}^{\text{after}}$ by adding back their interior connectivity. By Lemma~\ref{lem:non-overlapping}, the convexity of the boundary polyhedron guarantees that these embeddings are non-overlapping and non-inverted. For $\mathcal{P}^{\text{before}}$, we additionally need to place vertex $i$ inside face $f_k$. A simple choice is the barycenter of $f_k$:
\[
  \mathbf{p}_{i}= \frac{1}{|f_{k}|}\sum_{v \in f_k}\mathbf{p}_{v},
\]
where $|f_{k}|$ denotes the number of vertices in face $f_{k}$ (see Fig.~\ref{fig:initial-embedding}).

To obtain a low-distortion embedding, we perform joint optimization over both $\mathcal{P}
^{\text{before}}$ and $\mathcal{P}^{\text{after}}$ simultaneously. We minimize
the AMIPS~\cite{Fu2015} distortion energy:
\[
  \min_{\{\mathbf{p}_v\}}\quad E_{\text{AMIPS}}(\mathcal{P}^{\text{before}}) + E_{\text{AMIPS}}
  (\mathcal{P}^{\text{after}}),
\]
subject to the \textbf{coplanarity constraint}: vertex $i$ remains on the plane
containing face $f_{k}$. Specifically, let $\mathbf{n}_{k}$ be the normal vector
of face $f_{k}$, and let $v_{0}\in f_{k}$ be any vertex on $f_{k}$. The
constraint is:
\[
  (\mathbf{p}_{i}- \mathbf{p}_{v_0}) \cdot \mathbf{n}_{k}= 0.
\]

The joint optimization framework, combined with the coplanarity constraint,
yields a bijective mapping between $\mathcal{P}^{\text{before}}$ and
$\mathcal{P}^{\text{after}}$ with minimal distortion, thus completing the construction
of the local atlas for the boundary edge collapse operation. Note that the
shared scaffold structure can naturally extend to this 3D setting as well.

\subsection{Tracking Using Local Atlases}
\label{subsec:local-atlas-tracking}

Having constructed local bijective atlases for each type of remeshing operation, we can now
describe how to leverage these atlases to track geometric entities—points, curves,
and surfaces—throughout the remeshing sequence. A key advantage of our bijective approach is
that the local maps enable \emph{bidirectional tracking}:
we can map entities forward from the original mesh $\mathcal{M}$ to the remeshed
result $\mathcal{M}'$, or backward from $\mathcal{M}'$ to $\mathcal{M}$. For clarity
of exposition, we only consider the backward tracking perspective throughout this
section, though the forward direction follows by symmetry.

\subsubsection{Point Tracking}
\label{subsubsec:point-tracking}

\parheading{Point Representation}
A point $p$ on a simplicial mesh is represented by its barycentric coordinates with respect to a simplex:
\begin{equation}
  \label{eq:point-representation}p = (t, \mathbf{b}),
\end{equation}
where $t$ is the ID of a triangle in a triangle mesh or a tetrahedron in a tetrahedral mesh
$\mathbf{b}= (b_{0}, b_{1}, \ldots, b_{d})$ are the barycentric coordinates with
respect to the vertices of $t$ ($d = 2$ for triangles, $d = 3$ for tetrahedra).

\parheading{Backward Tracking Algorithm}
Consider a remeshing operation $o_{i}: \mathcal{M}_{i-1}\to \mathcal{M}_{i}$
with its associated local atlas
$(U_{i}, \mathcal{P}_{i}^{\text{before}}, \mathcal{P}_{i}^{\text{after}})$, where $U_i$ is the shared parametric domain,
$\mathcal{P}_{i}^{\text{before}}\subset \mathcal{M}_{i-1}$ and $\mathcal{P}_{i}^{\text{after}}
\subset \mathcal{M}_{i}$ are the local patches affected by the operation, and
$U_{i}$ is the shared parametric domain into which both patches are embedded.

Now let us consider where to map a point $p' = (t', \mathbf{b}')$ on $\mathcal{M}_{i}$ to $\mathcal{M}_{i-1}$ given operation $o_i$..
If it lies outside of $\mathcal{P}_i^{\text{after}}$ then it will not move, so its backtracked position is the trivially the same, i.e $p = p'$.
On the other hand, we first then map $p'$ to the parametric domain using barycentric interpolation
\begin{equation}
  \label{eq:bary-interp}
  \mathbf{u}' = \sum_{j=0}^{d} b'_{j}\mathbf{u}_{v_j},
\end{equation}
where $\mathbf{u}_{v_j}\in U_{i}$ are the parametric coordinates of the vertices of $t'$.

  We then locate the simplex $t \in \mathcal{P}_{i}^{\text{before}}$ containing $\mathbf{u}
    '$ and compute its barycentric coordinates $\mathbf{b}$ in $t$ to obtain the result $p = (t, \mathbf{b})$.

%
%

The correctness follows from the fact that both
$\mathcal{P}_{i}^{\text{before}}$ and $\mathcal{P}_{i}^{\text{after}}$ share the
same geometric boundary and are embedded into the same parametric domain $U_{i}$,
ensuring that any point $\mathbf{u}' \in U_{i}$ has a well-defined interpretation
in both mesh states.

\begin{figure}[t]
  \centering
  \includegraphics[width=0.95\linewidth]{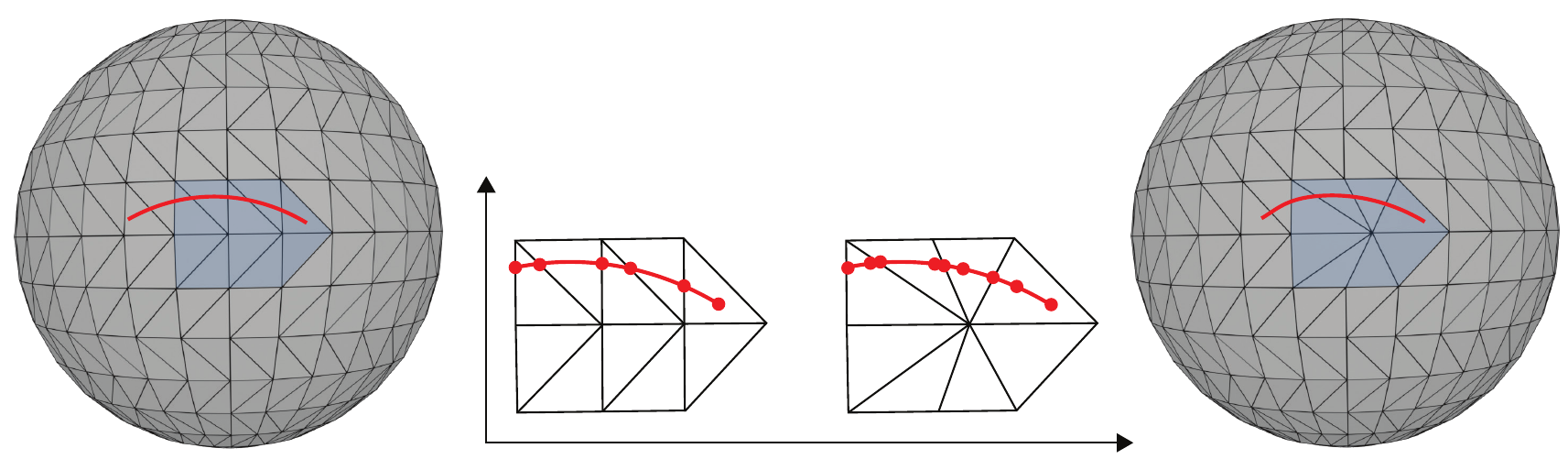}
  \caption{Curve tracking via local atlas. The portion of the curve (red) on the
  patch is first mapped to the local atlas (middle), where it is subdivided
  through arrangement with the new connectivity to generate intersection
  endpoints. This yields the tracked curve on the output mesh (right) with preserved
  topology.}
  \label{fig:curve_tracking}
\end{figure}

\subsubsection{Curve Tracking}
\label{subsubsec:curve-tracking}
\parheading{Curve Representation}
The task of tracking piecewise-linear curves on $\mathcal{M}'$ is not as trivial as simply projecting the endpoints of those curves to $\mathcal{M}$ because a single linear segment might not lie on $\mathcal{M}$.
In general, when the linear segments lie on more than one curve their embedding in 3D ceases to unambiguously lie on a mesh.
As such, the inputs and outputs of our curve tracking procedure are piecewise linear curves are comprised of linear segments that each lie on a single simplex (triangle or tetrahedron), though multiple segments may exist in a single simplex.
A curve $C$ on a simplicial mesh is therefore represented as a sequence of segments:
\begin{equation}
  \label{eq:curve-segments}C = \{\text{seg}_{1}, \text{seg}_{2}, \ldots, \text{seg}
  _{k}\},
\end{equation}
where each segment $\text{seg}_{i}$ is defined within a single simplex:
\begin{equation}
  \label{eq:segment-definition}\text{seg}_{i}= (t_{i}, \mathbf{b}_{i}^{(1)}, \mathbf{b}
  _{i}^{(2)}),
\end{equation}
with $t_{i}$ being the simplex  and $\mathbf{b}_{i}^{(1)}
, \mathbf{b}_{i}^{(2)}$ the barycentric coordinates of the two endpoints. We
require that consecutive segments are connected: the second endpoint of $\text{seg}
_{i}$ coincides with the first endpoint of $\text{seg}_{i+1}$, i.e.,
$(t_{i}, \mathbf{b}_{i}^{(2)}) \equiv (t_{i+1}, \mathbf{b}_{i+1}^{(1)})$ as points on
the mesh.
We use $\equiv$ because when segments end on simplex boundaries there are multiple triangles (respectively, tetrahedra) for which to construct barycentric coordinates from and we only care that the segment endpoints hold equivalent barycentric coordinates.

\parheading{Backward Tracking Problem}
Given a curve $C' = \{\text{seg}'_{1}, \ldots, \text{seg}'_{m}\}$ on $\mathcal{M}_{i}$ and a local atlas $(U_{i}, \mathcal{P}_{i}^{\text{before}}, \mathcal{P}_{i}^{\text{after}})$, we seek to compute the pre-image curve
$C = \{\text{seg}_{1}, \ldots, \text{seg}_{n}\}$ on $\mathcal{M}_{i-1}$. Since the
connectivity changes within the local patch, a segment in $C'$ that lies
entirely within a single simplex in $\mathcal{P}_{i}^{\text{after}}$ may
intersect multiple simplices in $\mathcal{P}_{i}^{\text{before}}$ due to the
different mesh topology.
As such, a single segment in $C'$ may result in multiple segments in $C$ to guarantee that each segment of $C$ lies in a single simplex in $\mathcal{M}$.

\parheading{Segment Arrangement Algorithm}
We perform our segment arrangement in the joint parametric space, so for a given segment $(t', \mathbf{b}'^{(1)}, \mathbf{b}'^{(2)})$ with
$t' \in \mathcal{P}_{i}^{\text{after}}$, we map its endpoints to the parametric domain via~\eqref{eq:bary-interp}, yielding parametric positions $\mathbf{u}^{(1)}$ and $\mathbf{u}^{(2)}$.

Our goal is to compute the arrangement of the segment $[\mathbf{u}^{(1)}, \mathbf{u}
^{(2)}]$ with respect to $\mathcal{P}_{i}^{\text{before}}$ in the joint parametric space $U_i$.
To do this employ a ray-marching: starting from $\mathbf{u}^{(1)}$, we trace along
direction $\mathbf{k}= \mathbf{u}^{(2)}- \mathbf{u}^{(1)}$ until reaching $\mathbf{u}
^{(2)}$, computing intersections with simplex boundaries (edges for 2D, faces
for 3D) in $\mathcal{P}_{i}^{\text{before}}$ at each step.
These intersection points will all lie on simplex boundaries, yielding sub-segments $\{\widetilde{\text{seg}}
_{1}, \ldots, \widetilde{\text{seg}}_{\ell}\}$ where each lies entirely within a
single simplex of $\mathcal{P}_{i}^{\text{before}}$ (Figure~\ref{fig:curve_tracking}).

\parheading{Intersection Candidate Selection}
At each step, given the current position $\mathbf{V}$ in the parametric domain and
direction $\mathbf{k}$, we determine candidate facets (edges for 2D, faces for
3D) based on the location of $\mathbf{V}$ (Figure~\ref{fig:candidate-selection}).

\begin{figure}[t]
  \centering
  \includegraphics[width=0.85\linewidth]{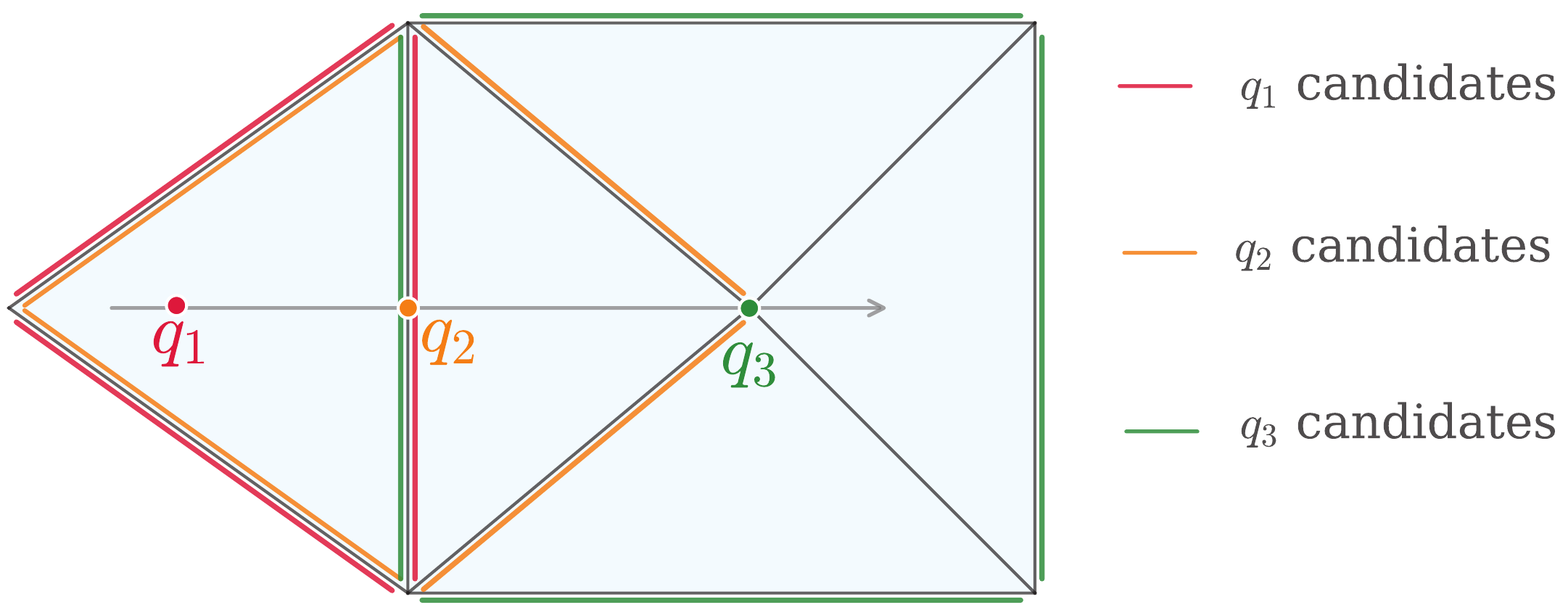}
  \caption{Candidate facet selection during curve tracking: Given
  three query points $q_{1}$, $q_{2}$, $q_{3}$ along a ray, we select candidate
  edges (highlighted in color) based on each point's location. Point $q_{1}$ lies
  in the interior of a triangle, so all three edges are candidates. Point $q_{2}$
  lies on an edge, so candidates come from the link of that edge (opposite edges
  from incident triangles). Point $q_{3}$ lies at a vertex, so candidates are
  edges opposite to that vertex in all incident triangles. The algorithm selects
  the intersection with the smallest positive parameter $t$ to advance to the next
  point.}
  \label{fig:candidate-selection}
\end{figure}

If $\mathbf{V}$ lies strictly inside a simplex $t$, the candidates are all
$(d-1)$-dimensional facets of $t$: three edges for triangles, four faces for
tetrahedra.

If $\mathbf{V}$ lies on a lower-dimensional sub-simplex $\sigma$ (a vertex, an edge for 2D/3D, or a face for 3D), we collect
candidates from the link~\citep{Munkres_2018} of $\sigma$, i.e all $(d-1)$-dimensional facets
from simplices incident to $\sigma$ excluding those facets that contain $\sigma$
itself.

\begin{itemize}
\item
\textbf{In 2D:} if $\sigma$ is an edge $(i,j)$, then $\text{link}((i,j))$ includes the
opposite edges from all triangles sharing $(i,j)$. If $\sigma$ is a vertex
$v$, then $\text{link}(v)$ includes all edges opposite to $v$ in triangles containing
$v$.
\item
\textbf{In 3D:} if $\sigma$ is a face, then $\text{link}(\sigma)$ includes the opposite
faces from tetrahedra sharing that face. If $\sigma$ is an edge, then $\text{link}
(\sigma)$ includes all faces from incident tetrahedra that do not contain the
edge. If $\sigma$ is a vertex $v$, then $\text{link}(v)$ includes all faces opposite
to $v$ in tetrahedra containing $v$.
\end{itemize}

\parheading{Ray-Facet Intersection}
For each candidate facet $f$, we compute the ray-facet intersection (ray-edge
for 2D, ray-triangle for 3D) and select the one with smallest positive parameter
$t > 0$ as the next intersection point. The detailed formulation is provided in Appendix~\ref{appendix:ray-intersection}.
This process repeats until reaching $\mathbf{u}^{(2)}$, producing the sub-segments
$\{\widetilde{\text{seg}}_{1}, \ldots, \widetilde{\text{seg}}_{\ell}\}$ in $\mathcal{P}
_{i}^{\text{before}}$.

\subsubsection{Topology-Preserving Multi-Curve Tracking}
\label{subsubsec:topology-preserving-curve-tracking}

In applications such as tracking UV seams or material boundaries on triangle meshes,
we will need to track multiple curves or loops simultaneously. In such scenarios,
the primary requirement is preserving the topological relationships among
curves.
Specifically, this means maintaining their intersection and overlap patterns in the correct
order.
Although our maps are bijective, running our procedure using floating point arithmetic occasionally changes the topological relationships between curves.
On the other hand, exact rational arithmetic becomes computationally prohibitive.
To balance these two extremes we take the approach of defining invariants that can help us determine whether a floating point computation introduced topological changes, and if a topological change does occur we re-run our tracking algorithm using exact rational predicates.

\parheading{Intersection and Overlap Events}
Consider a set of curves $\mathcal{C}= \{C_{1}, \ldots, C_{\ell}\}$ to be
tracked. As we traverse a curve $C_{i}$ from start to end, it encounters a
sequence of events where it interacts with other curves (including itself for self-intersections).
We classify these events into two types:

\begin{itemize}
  \item \textbf{Intersection}: Two segments cross at a single point (either at
    shared endpoints or at a transversal intersection within a common simplex).

  \item \textbf{Overlap}: Two segments coincide along a positive-length interval
    (not just a single point).
\end{itemize}

For each curve $C_{i}$, we record the ordered sequence of events:
\begin{equation}
  \mathcal{E}(C_{i}) = \langle e_{1}, e_{2}, \ldots, e_{m}\rangle,
\end{equation}
where each event $e_{k}$ specifies:
\begin{itemize}
  \item The type (intersection or overlap),

  \item The interacting curve index $j$ (where $j$ may equal $i$ for self-interactions),

  \item The parametric location(s) along $C_{i}$ where the event occurs.
\end{itemize}

\parheading{Topology Preservation Requirement}
Our goal is to ensure that for each curve $C_{i}\in \mathcal{C}$ and its tracked
result $C_{i}'$, the event sequence is preserved:
\begin{equation}
  \mathcal{E}(C_{i}) \cong \mathcal{E}(C_{i}'),
\end{equation}
where $\cong$ denotes combinatorial equivalence: the sequences have the same
length, and corresponding events have the same type and involve the same curve pairs.

Since intersections and overlaps are defined intrinsically within each simplex,
we verify topology preservation locally for each operation. After tracking all curves
through the local atlas $(U_{i}, \mathcal{P}_{i}^{\text{before}}, \mathcal{P}_{i}
^{\text{after}})$, we compute the event sequences in
$\mathcal{P}_{i}^{\text{before}}$ and verify that they match the original
sequences in $\mathcal{P}_{i}^{\text{after}}$.

\parheading{Implementation Strategy}
We track all curves within a local patch simultaneously and extract
their event sequences by:
\begin{enumerate}
  \item For each pair of curves $(C_{i}, C_{j})$ (including $i = j$), enumerate
    all segment pairs and classify their interactions.

  \item Sort events along each curve according to their parametric positions.

  \item Compare the resulting event sequences before and after tracking.
\end{enumerate}

If the event sequences do not match ,typically due to numerical errors in double-precision
arithmetic, we re-run tracking using exact rational arithmetic to ensure
correctness.

\parheading{Curve Simplification via Edge Collapse}
To improve efficiency when tracking many curves over long remeshing sequences, we
simplify curves by collapsing short segments. Given a segment
$\text{seg}= (t, \mathbf{b}^{(1)}, \mathbf{b}^{(2)})$ with parametric length
below a threshold $\epsilon$, we consider collapsing it by merging its endpoints.

Crucially, edge collapse must preserve the event sequences $\mathcal{E}(C_{i})$
for all curves. Before collapsing a segment in curve $C_{i}$, we verify that
  the local event sequence near the segment (including events involving
    $C_{i}$ and other curves) remains unchanged after collapse and  that
  no new intersections or overlaps are created, and no existing ones are
    destroyed.

This verification can be performed by computing event sequences in a local
neighborhood before and after each potential collapse. By iteratively simplifying
curves while preserving their event sequences, we significantly reduce
computational cost without compromising topological correctness.

\subsubsection{Surface Tracking on Tetrahedral Meshes}
\label{subsubsec:surface-tracking-tet}
For tetrahedral meshes, we extend the tracking framework to handle surfaces, which
are essential for applications such as tracking material interfaces or segmentation
boundaries in volumetric simulations.

\parheading{Surface Representation}
A surface $S$ on a tetrahedral mesh is represented as a triangle mesh embedded
in the volume, where each vertex is a point on the tetrahedral mesh:
\begin{equation}
  p = (t, \mathbf{b}),
\end{equation}
with $t$ being a simplex(tetrahedron) ID and $\mathbf{b}= (b_{0}, b_{1}, b_{2}, b
_{3})$ the barycentric coordinates. The surface consists of triangular faces:
\begin{equation}
  S = \{f_{1}, f_{2}, \ldots, f_{k}\},
\end{equation}
where each face $f_{i}= (p_{i}^{(1)}, p_{i}^{(2)}, p_{i}^{(3)})$ is defined by three
vertices. We require that all three vertices of each face lie within the same
tetrahedron, including its boundary. That is, if $p_{i}^{(j)}= (t_{i}^{(j)}, \mathbf{b}
_{i}^{(j)})$ for $j \in \{1, 2, 3\}$, then either
$t_{i}^{(1)}= t_{i}^{(2)}= t_{i}^{(3)}$, or they share a common face, edge, or vertex
of the tetrahedral mesh.

\parheading{Surface Arrangement Algorithm}
Given a surface $S'$ on $\mathcal{M}_{i}$ and a local atlas
$(U_{i}, \mathcal{P}_{i}^{\text{before}}, \mathcal{P}_{i}^{\text{after}})$, we track
the portion of $S'$ that intersects the local patch $\mathcal{P}_{i}^{\text{after}}$.

For each face $f \in S'$ with $f \subset \mathcal{P}_{i}^{\text{after}}$, we:
\begin{enumerate}
  \item Map the three vertices of $f$ to the parametric domain $U_{i}$,
    obtaining a triangle $\tilde{f}$ in the parametric domain.

  \item Extract the boundary surface of $\mathcal{P}_{i}^{\text{before}}$ restricted
    to the region covered by $\tilde{f}$.

  \item Compute the arrangement of $\tilde{f}$ with the boundary triangles of $\mathcal{P}
    _{i}^{\text{before}}$ using the \texttt{autorefine\_triangle\_soup} function
    from CGAL~\cite{cgal:lty-pmp-25b}, which subdivides overlapping triangles into
    non-overlapping pieces and retriangulates the result.

  \item Convert the resulting triangles back to barycentric coordinates with respect
    to tetrahedra in $\mathcal{P}_{i}^{\text{before}}$.
\end{enumerate}

The CGAL arrangement procedure~\citep{cgal:eb-25b} guarantees that the output is a valid
triangulation where each triangle lies entirely within a single tetrahedron or
on its boundary. To ensure robustness, this computation can be performed using
exact rational arithmetic.

\parheading{Topology Preservation for Multiple Surfaces}
When tracking multiple surfaces $\mathcal{S}= \{S_{1}, \ldots, S_{m}\}$ simultaneously,
we preserve their topological relationships through an intersection structure
between intersecting surfaces.

Just like with curve tracking, we will guarantee topology by developing an invariant.
For surfaces, the relevant topological invariant is the number and connectivity
of intersection curves. If surfaces $S_{i}$ and $S_{j}$ intersect, their
intersection forms a set of curves (possibly multiple disjoint components). We define
the \emph{intersection graph} $\mathcal{G}(\mathcal{S})$ where:
\begin{itemize}
  \item Nodes represent connected components of pairwise surface intersections.

  \item Edges connect components that share endpoints or form junctions.
\end{itemize}

Our goal is to ensure that, for the tracked surfaces $\mathcal{S}'$, the intersection
graph remains combinatorially equivalent:
$\mathcal{G}(\mathcal{S}) \cong \mathcal{G}(\mathcal{S}')$.

\parheading{Surface Simplification via Edge Collapse}
Similar to curve tracking, we simplify surface representations by collapsing
short edges to improve efficiency. Given an edge $e = (p^{(1)}, p^{(2)})$ in a surface,
we collapse it if its length is below a threshold $\epsilon$, if the collapse keeps all incident triangles remain within a single tetrahedron or its boundary, and if the intersection graph $\mathcal{G}(\mathcal{S})$, is preserved.

To verify topology preservation, we compute the number of connected components in
the intersection curves before and after collapse. If the component count
changes, we reject the collapse. By iteratively applying valid collapses, we maintain
a simplified surface representation throughout the tracking process without compromising
topological correctness.

\section{Results}
\label{sec:Results}

\subsection{Texture Transfer}

A practical application of our point tracking framework is texture transfer
across remeshed models (\Cref{fig:texture_transfer_spot,fig:texture_transfer_ogre}).
Given an input mesh with an existing UV parameterization and texture, our goal
is to generate equivalent textures for remeshed versions of the model that may
use entirely different parameterizations.

The key insight is that our bijective mapping decouples remeshing from UV
maintenance. Traditional approaches must carefully preserve UV coordinates throughout
remeshing operations, which becomes increasingly difficult when the mesh
contains complex UV seams or multiple charts\cite{sander2001texture}. In
contrast, our method allows the remeshed model to be reparameterized independently
of any choice in standard parameterization algorithm. The texture is then
transferred by composing the new parameterization with our tracked bijective
correspondence: we evalute each texel in the new texture domain by first mapping
to the remeshed surface, tracking it backward to the original mesh, and finally sample
the original texture at the corresponding UV coordinates.

\Cref{fig:texture_transfer_spot,fig:texture_transfer_ogre} demonstrate texture transfer
through both mesh decimation and refinement operations. The center images show
the original models with their input UV parameterizations and textures. The top row
displays the remeshed results after decimation (left) and refinement (right),
both rendered with successfully transferred textures. The bottom row shows the UV
layouts and corresponding texture maps generated for each remeshed model.
Despite significant changes in mesh resolution and connectivity our bijective
framework ensures accurate texture correspondence without introducing artifacts
such as seam tearing or texture deviation.

\begin{figure}[htb]
    \centering
    \includegraphics[width=0.99\linewidth]{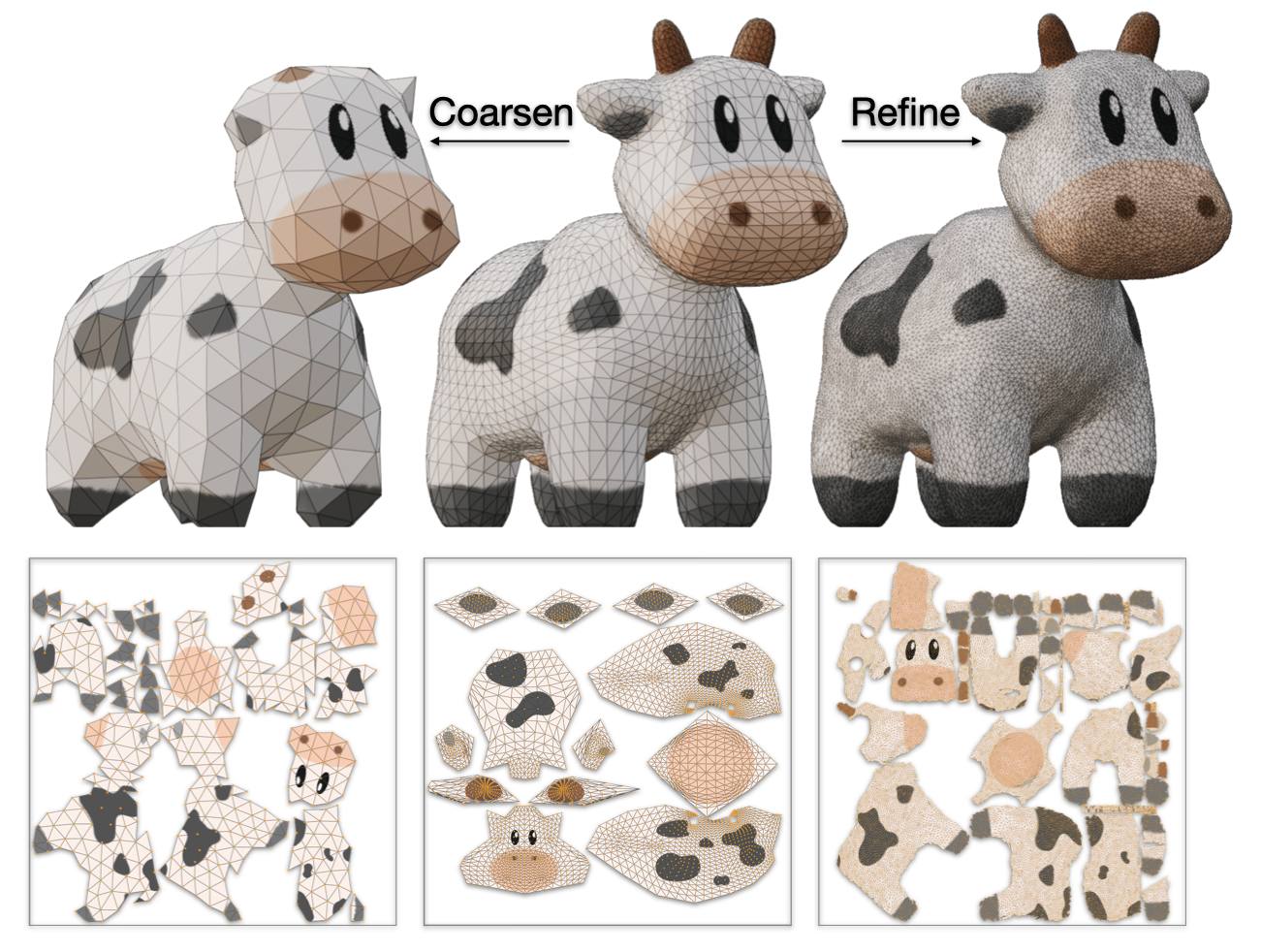}
    \caption{Texture transfer through mesh decimation and refinement. \textbf{Center}:
    the input spotted animal model\cite{crane2013robust} with its original UV
    parameterization and texture. \textbf{Top row}: remeshed results after decimation
    (left) and refinement (right), rendered with the transferred texture. \textbf{Bottom
    row}: the UV layouts and generated textures corresponding to fresh parameterizations
    of each remeshed model. Our bijective mapping enables texture transfer by
    composing the new parameterization with the tracked correspondence, eliminating
    the need to maintain UV coordinates during remeshing operations.}
    \label{fig:texture_transfer_spot}
\end{figure}

\begin{figure}[htb]
    \centering
    \includegraphics[width=0.99\linewidth]{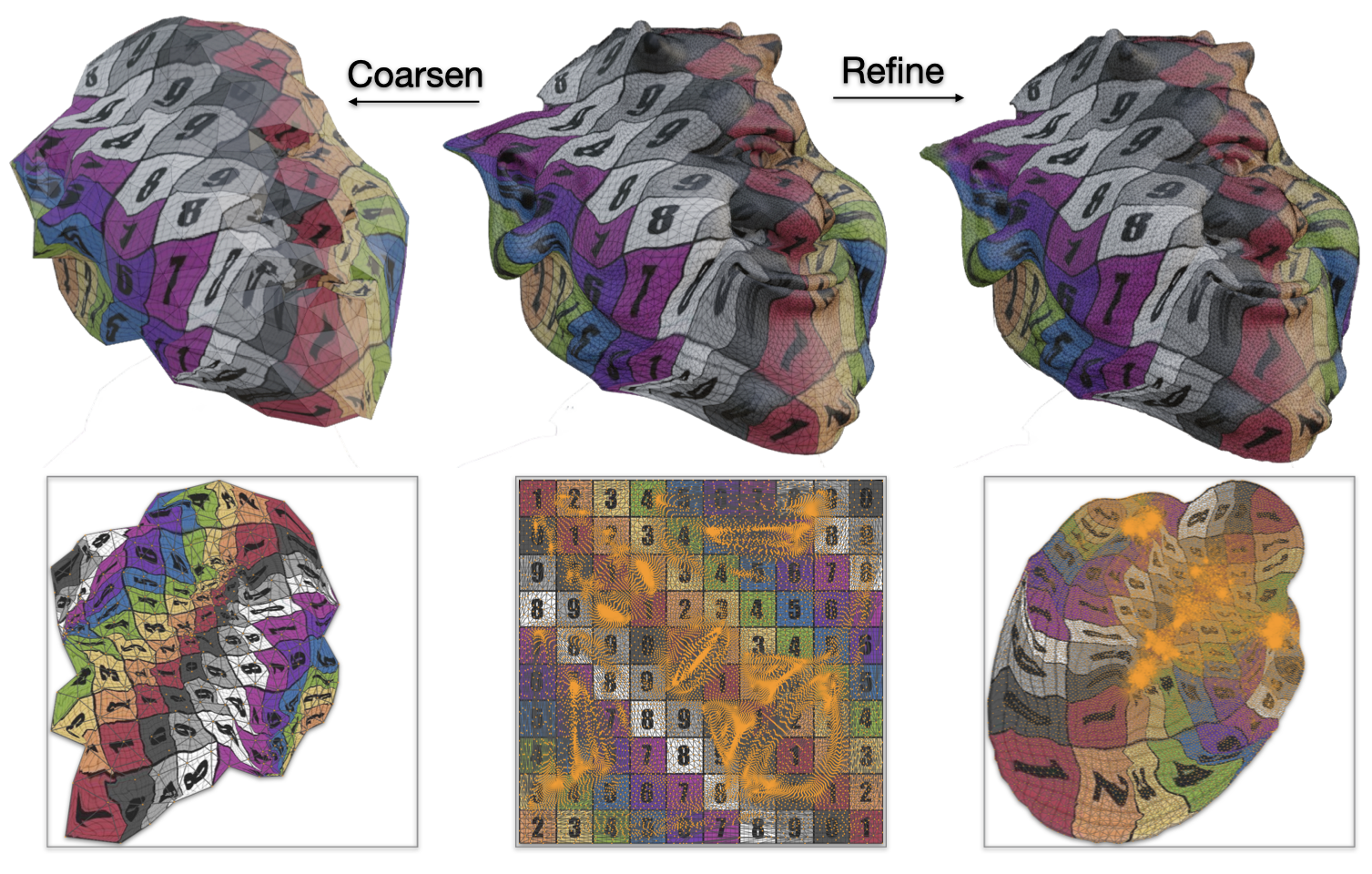}
    \caption{Texture transfer on the Ogre model with checkerboard pattern. The regular
    grid structure of the checkerboard provides a visual metric for evaluating mapping
    quality and distortion. Despite extensive geometric and connectivity changes
    from decimation (top left) and refinement (top right), the checkerboard pattern
    remains smooth and continuous, demonstrating effective distortion control
    through our bijective framework.}
    \label{fig:texture_transfer_ogre}
\end{figure}

\subsection{Curve Tracking on Triangle Meshes.}
We demonstrate our curve tracking algorithm on triangle meshes undergoing
remeshing operations (Figure~\ref{fig:curve_tracking_result}). Given an input mesh
$\mathcal{M}_{\text{input}}$, we generate a collection of curves by intersecting
the surface with axis-aligned planes parallel to the $xy$, $yz$, and $xz$
coordinate planes. These curves form a complex network with several intersection
points where curves from different plane families cross each other.

We then perform forward tracking through the remeshing sequence to obtain the corresponding
curves on $\mathcal{M}_{\text{output}}$. As shown in Figure~\ref{fig:curve_tracking_result},
the tracked curves faithfully preserve their intersection topology: the
combinatorial pattern of curve crossings remains identical between input and
output, with no spurious intersections introduced and no existing intersections lost.
This topological guarantee is achieved through our topology-preserving multi-curve
tracking algorithm (Section~\ref{subsubsec:topology-preserving-curve-tracking}),
which maintains ordered sequences of intersection events throughout the tracking
process. Critically, during the segment arrangement phase, we employ exact
arithmetic with coordinates represented as rational numbers to evaluate
geometric predicates, ensuring numerical robustness and preventing topological inconsistencies
that could arise from floating-point rounding errors.

\begin{figure}[hbt]
    \centering
    \includegraphics[width=0.99\linewidth]{
        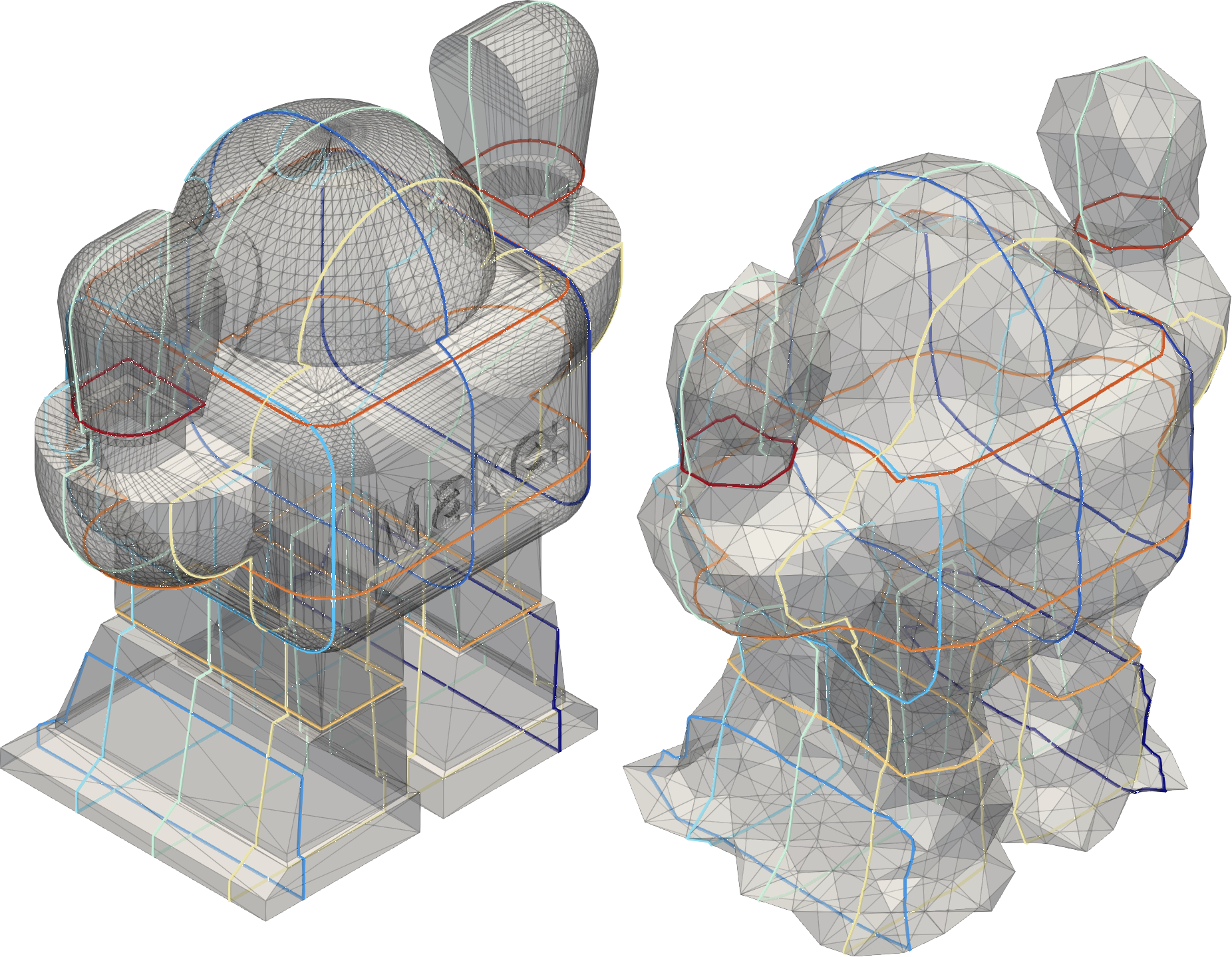
    }
    \caption{Curve tracking on a triangle mesh. \textbf{Left}: curves on the input
    mesh $\mathcal{M}_{\text{input}}$, obtained as intersections of the surface
    with axis-aligned planes (parallel to the $xy$, $yz$, and $xz$ planes), shown
    in distinct colors. \textbf{Right}: the forward-tracked curves on the output
    mesh $\mathcal{M}_{\text{output}}$ after isotropic remeshing. Our method preserves
    the intersection topology of the curves throughout the remeshing sequence—curves
    that intersect on the input mesh maintain their intersection relationships on
    the output mesh.}
    \label{fig:curve_tracking_result}
\end{figure}

We evaluated our curve tracking algorithm on 5,139 manifold triangle meshes from
the Thingi10K dataset~\cite{Zhou2016}. For each model, we performed isotropic
remeshing and tracked axis-aligned planar curves through the operation sequence.
Of these models, 4,998 (97.3\%) completed successfully with topology preservation
verified throughout. The remaining 141 models timed out after 12 hours, primarily
due to extensive operation counts exceeding 700,000 per model. These timeouts
represent computational resource limits rather than algorithmic failures—given
sufficient computation time, they would yield topologically correct results identical
to the successful cases. Furthermore, as we note in \Cref{sec:concl}, the
majority of this cost is due to computing each atlas as part of each topological
operation. In reality this could be performed afterwards and in parallel.

Figure~\ref{fig:stress_test} showcases stress test examples where meshes undergo
aggressive simplification with dramatic geometric changes. Despite surface
features being substantially altered or eliminated, our method maintains topologically
correct curve tracking—intersection patterns and connectivity remain intact.
This demonstrates that our bijective framework preserves topological coherence
even under severe geometric deformations: as long as the remeshing operations are
topologically valid, our composed mappings produce correspondingly coherent
results.

\begin{figure*}[t]
    \centering
    \includegraphics[width=0.98\linewidth]{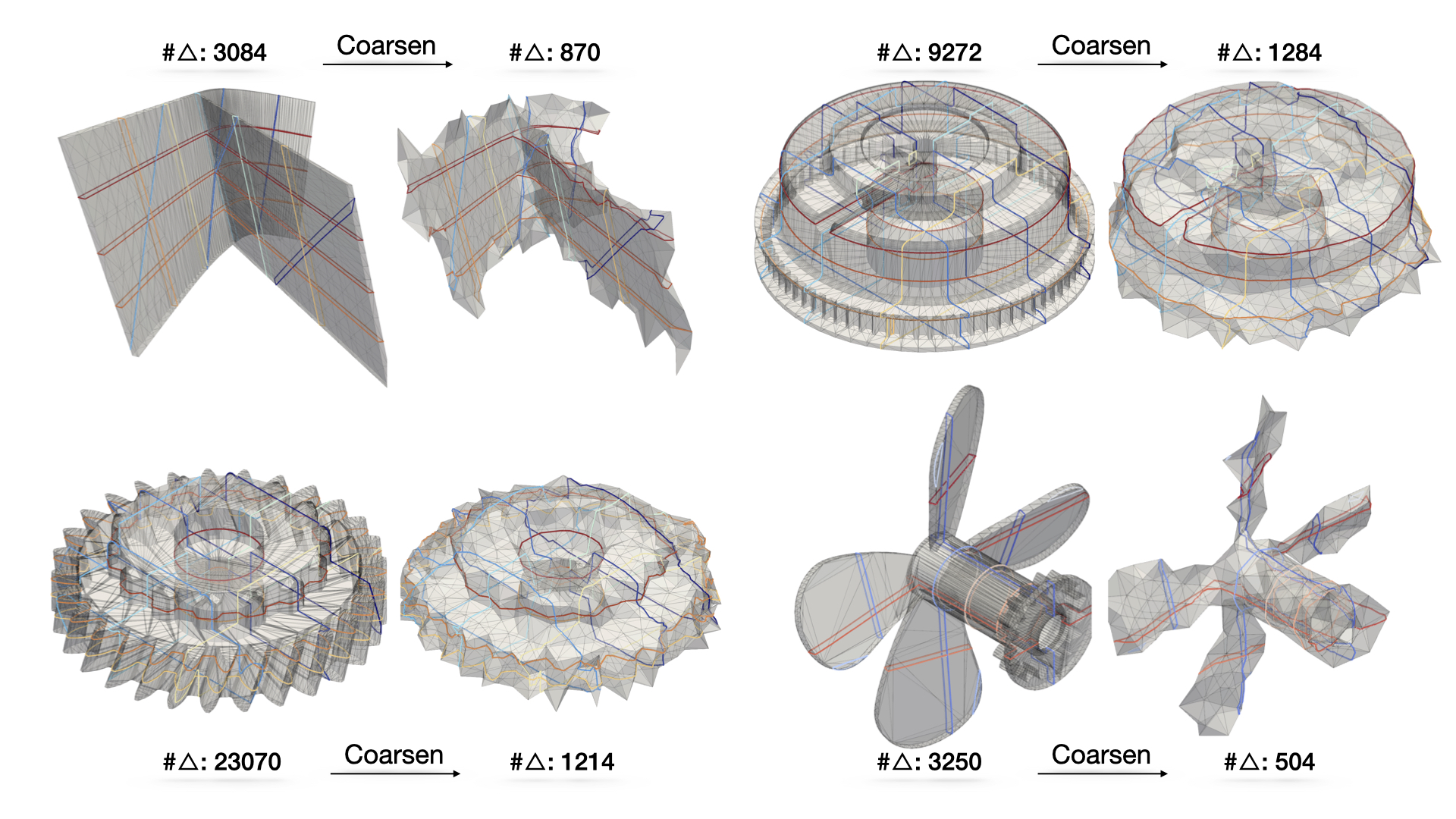}
    \caption{Stress test for curve tracking under extreme remeshing. Each pair shows
    input mesh (left) and aggressively simplified output (right) with tracked curves
    in color. Despite dramatic geometric changes, our method preserves topological
    correctness—curve intersections and connectivity remain intact throughout.
    Thingi10K models: \#636811 (top-left), \#1036656 (top-right), \#1312974 (bottom-left),
    \#1505135 (bottom-right).}
    \label{fig:stress_test}
\end{figure*}


\subsection{Surface Tracking on Tetrahedral Meshes}
We demonstrate the application of our surface tracking framework to medical
imaging data, specifically tracking organ segmentation boundaries through volumetric
mesh simplification (Figure~\ref{fig:ct_tracking_result}).

\begin{figure}[t]
    \centering
    \includegraphics[width=0.99\linewidth]{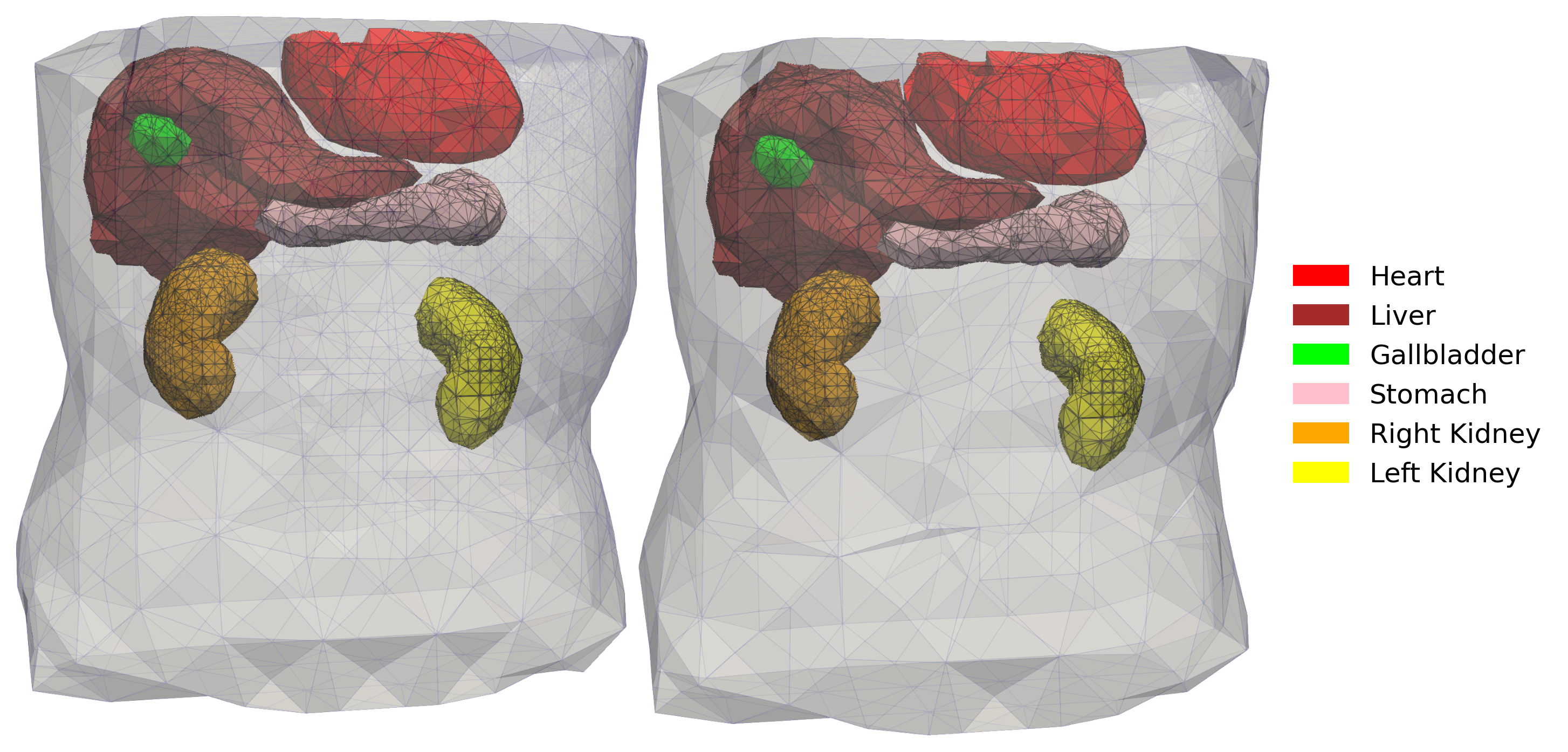}
    \caption{Surface tracking on a tetrahedral body mesh from CT scan data. We track
    multiple organ segmentation surfaces (heart, liver, gallbladder, stomach,
    and kidneys) through volumetric mesh simplification. \textbf{Left}: organ
    surfaces on the input tetrahedral mesh $\mathcal{M}_{\text{input}}$. \textbf{Right}:
    the tracked surfaces on the simplified output mesh $\mathcal{M}_{\text{output}}$.
    Our bijective framework maintains the topological relationships among organs
    throughout the remeshing process.}
    \label{fig:ct_tracking_result}
\end{figure}

Our input data is derived from the CTA Abdomen (Panoramix) sample dataset
provided in 3D Slicer. We obtain organ segmentations using TotalSegmentator~\cite{wasserthal2023totalsegmentator},
a deep learning-based tool for automatic whole-body CT segmentation. We extract the
surface triangulation of six major organs (heart, liver, gallbladder, stomach,
and kidneys) as separate surfaces to be tracked, with each surface embedded in
the volumetric mesh via barycentric coordinates.

We perform mesh simplification on the background body mesh while tracking the
organ surfaces through each local atlas. As shown in Figure~\ref{fig:ct_tracking_result},
the tracked organ surfaces maintain their geometric shapes and relative
positions after simplification. The topological relationships among organs are preserved:
surfaces that were previously disjoint remain non-intersecting throughout the
remeshing sequence.

To evaluate surface tracking on a broader range of geometric models, we perform large-scale
testing on tetrahedral meshes generated from the Thingi10K dataset~\cite{Zhou2016}
using TetWild~\cite{Hu:2018:TMW:3197517.3201353} (Figure~\ref{fig:teaser}). For
each input model, we first perform mesh simplification to obtain a coarsened output
mesh $\mathcal{M}_{\text{output}}$. We then sample axis-aligned planar surfaces
on this simplified mesh by intersecting it with planes parallel to the $xy$, $xz$,
and $yz$ coordinate planes.

We perform backward tracking to map these surfaces from
$\mathcal{M}_{\text{output}}$ to the original high-resolution mesh
$\mathcal{M}_{\text{input}}$. Critically, the intersection topology is preserved
throughout the tracking process: the intersection graph $\mathcal{G}(\mathcal{S})$
representing how surfaces intersect remains combinatorially equivalent. As shown
in Figure~\ref{fig:teaser}, surfaces that intersect on the simplified output mesh
maintain their intersection relationships when back-tracked to the input mesh,
with intersection curves faithfully reconstructed.

\section{Conclusions}
\label{sec:concl}

We have presented BijectiveRemesh, a robust framework for maintaining bijective mappings throughout complex remeshing sequences on both 2D triangle meshes and 3D tetrahedral meshes. Our approach guarantees global bijectivity through local atlas construction using two key innovations: shared scaffold structures for 2D operations and convex polyhedra embeddings for 3D boundary operations. These bijective mappings enable exact tracking of geometric entities—points, curves, and surfaces—with rigorous preservation of topological relationships, eliminating the artifacts common in projection-based transfer methods.

\textbf{Limitations and Future Work.} Constructing the bijective local atlases adds approximately $110\times$ overhead per operation compared to performing remeshing operations alone. 
This overhead reflects the fact that our prototype augments each local patch with the auxiliary triangulation and performs iterative energy minimization with inversion-preventing line searches. 
We believe a substantial proportion of this overhead can be reduced through parallelization. Our current prototype implementation processes operations serially—constructing each local atlas sequentially as operations are applied. In practice, the atlas construction for different operations is independent: we can first record the local patch information for all operations during remeshing, then construct their corresponding local atlases in parallel. Since each operation's atlas construction is self-contained, this parallel formulation would reduce total execution time.

We currently also depend on exact rational arithmetic for geometric predicates in our current tracking implementation for the sake of robustness. While we employ various optimizations including rounding to double and curve simplification, the computational overhead of performing long sequences of rational arithmetic remains substantial, making the tracking process can become prohibitively slow.

Our framework is designed for incremental, local-modification remeshing pipelines (edge splits, collapses, flips, and vertex smoothing). Global remeshing strategies—such as voxelization-based remeshing, Delaunay re-triangulation, or neural mesh extraction—do not decompose into sequences of local operations and are therefore outside the current scope. Extending bijective mapping maintenance to such non-incremental methods would require fundamentally different techniques and remains an open problem.

An important direction for future work is developing floating-point tracking algorithms that maintain topological guarantees without exact arithmetic. This would require novel geometric predicates and consistency checks that can tolerate numerical errors while still preventing topological corruption in the tracked curves and surfaces.


\bibliographystyle{ACM-Reference-Format}
\bibliography{99-biblio}

\appendix
\section{Appendix}
\label{appendix}

\subsection{Constructive Algorithms for Convex Polyhedra Embedding}
\label{appendix:convex-polyhedra}

The following two algorithms describe the constructive procedure for embedding a planar graph as a convex polyhedron, as referenced in Section~\ref{subsubsec:boundary-edge-collapse}. Algorithm~\ref{alg:tutte-embedding} computes a crossing-free planar embedding via Tutte's barycentric method, and Algorithm~\ref{alg:maxwell-cremona} lifts it into a 3D convex polyhedron via the Maxwell-Cremona correspondence.

\begin{algorithm}[H]
  \caption{Tutte Barycentric Embedding}
  \label{alg:tutte-embedding}
  \begin{algorithmic}
    [1] \Require Planar graph given by face list $F = \{f_{1}, f_{2}, \ldots\}$ and
    outer triangle $f_{0}= (v_{1}, v_{2}, v_{3})$ \Ensure 2D embedding $p: V \to
    \mathbb{R}^{2}$

    \Statex \textbf{Construct weighted Laplacian} \For{each undirected edge $(i,j)$}
    \State $\omega_{ij}\gets
    \begin{cases}
      1, & \text{if $(i,j)$ is an interior edge}, \\
      0, & \text{if }i, j \in f_{0}
    \end{cases}$ \EndFor \State Assemble the $n \times n$ weighted Laplacian: $L_{ii}
    = \sum_{j}\omega_{ij}$, $L_{ij}= -\omega_{ij}$ for $i \neq j$

    \Statex \textbf{Partition variables} \State Let $B$ be the three boundary
    vertex indices in $f_{0}$, and $I$ be the interior vertices \State Partition:
    $L =
    \begin{pmatrix}
      L_{BB} & L_{BI} \\
      L_{IB} & L_{II}
    \end{pmatrix}$

    \Statex \textbf{Fix boundary vertices} \State $p_{v_1}\gets (0, 0)$, $p_{v_2}
    \gets (1, 0)$, $p_{v_3}\gets (0, 1)$ \State $x_{B}\gets (0, 1, 0)^{\top}$, $y
    _{B}\gets (0, 0, 1)^{\top}$

    \Statex \textbf{Solve for interior vertices} \State Solve $L_{II}x_{I}= -L_{IB}
    x_{B}$ and $L_{II}y_{I}= -L_{IB}y_{B}$ \State \Return Crossing-free planar embedding
    $p: V \to \mathbb{R}^{2}$
  \end{algorithmic}
\end{algorithm}

\begin{algorithm}[H]
  \caption{Maxwell-Cremona Lifting to Convex Polyhedron}
  \label{alg:maxwell-cremona}
  \begin{algorithmic}
    [1] \Require Planar embedding $p: V \to \mathbb{R}^{2}$ from Algorithm~\ref{alg:tutte-embedding}
    \Ensure Convex polyhedra embedding $\{(p_{v}, z_{v}) \mid v \in V\} \subset \mathbb{R}
    ^{3}$

    \Statex \textbf{Assign equilibrium stress} \For{each interior edge $(i,j)$}
    \State $\omega_{ij}\gets 1$ \EndFor

    \Statex \textbf{Define plane per face} \Statex For each face $f_{k}$, seek plane
    $H_{k}: z = \langle (x, y), a_{k}\rangle + d_{k}$

    \Statex \textbf{Maxwell-Cremona propagation} \State Choose one interior face
    $f_{1}$ as base: $a_{1}\gets (0, 0)$, $d_{1}\gets 0$ \For{each pair of adjacent faces $f_{r}, f_{\ell}$ sharing edge $(i,j)$}
    \State where $f_{\ell}$ lies to the left of directed edge $i \to j$ \State
    $a_{\ell}\gets \omega_{ij}(p_{i}- p_{j})^{\perp}+ a_{r}$ \Comment{$(x, y)^{\perp}= (-y, x)$}
    \State $d_{\ell}\gets \omega_{ij}\langle p_{i}, (p_{j})^{\perp}\rangle + d_{r}$
    \EndFor

    \Statex \textbf{Compute vertex heights} \For{each vertex $v \in V$} \State
    Pick any incident face $f_{k}$ \State
    $z_{v}\gets \langle p_{v}, a_{k}\rangle + d_{k}$ \EndFor

    \State \Return Convex polyhedron $\{(p_{v}, z_{v})\}$ whose projection onto $x
    y$-plane is the Tutte embedding
  \end{algorithmic}
\end{algorithm}

\subsection{Non-overlapping via Convex Boundary Embedding}

\begin{lemma}[Non-overlapping via Convex Boundary Embedding]
\label{lem:non-overlapping}
Let $\mathcal{P}_{\text{before}}$ be the local patch of boundary tetrahedra incident to vertex $i$. The boundary of this patch consists of:
\begin{itemize}
    \item Face $f_k$: the one-ring of vertex $i$ on $\partial \mathcal{M}_{\ell}$ (containing vertex $i$),
    \item Faces $f_0, \ldots, f_{k-1}$: the $k$ adjacent boundary triangles.
\end{itemize}
Suppose the vertices of $\{f_0, \ldots, f_{k-1}, f_k\}$ (excluding $i$) are embedded to form a convex polyhedron $\mathcal{C}$ via Algorithms~\ref{alg:tutte-embedding} and~\ref{alg:maxwell-cremona}, where we select one triangle from $\{f_0, \ldots, f_{k-1}\}$ as the outer boundary for the Tutte embedding.

Then for any position of vertex $i$ on or inside face $f_k$, all tetrahedra in $\mathcal{P}_{\text{before}}$ are non-overlapping.
\end{lemma}

\begin{proof}
A tetrahedral mesh is non-overlapping if and only if for every internal face, the two opposite vertices lie on opposite sides of that face.

Consider an arbitrary internal face $f = (i, v_a, v_b)$ shared by two tetrahedra $T_1 = (i, v_a, v_b, v_c)$ and $T_2 = (i, v_a, v_b, v_d)$. We must show that $v_c$ and $v_d$ lie on opposite sides of the plane $\Pi$ containing $f$.

\textbf{Key observation.} The face $f = (i, v_a, v_b)$ contains the edge $(v_a, v_b)$, which is an edge of the convex polyhedron $\mathcal{C}$. The vertices $v_c$ and $v_d$ are the two boundary vertices adjacent to this edge, forming boundary faces $(v_a, v_b, v_c)$ and $(v_a, v_b, v_d)$ of $\mathcal{C}$.

\textbf{Convexity guarantee.} By the convexity of $\mathcal{C}$, the dihedral angle at edge $(v_a, v_b)$ is less than $\pi$. This means that for \emph{any} plane $\Pi$ containing edge $(v_a, v_b)$, the vertices $v_c$ and $v_d$ lie on opposite sides of $\Pi$.

In particular, since vertex $i$ is positioned on or inside face $f_k$ of the convex polyhedron $\mathcal{C}$, and the plane $\Pi$ through face $(i, v_a, v_b)$ contains edge $(v_a, v_b)$, we have:
\[
[(\mathbf{p}_c - \mathbf{p}_a) \cdot \mathbf{n}] \cdot [(\mathbf{p}_d - \mathbf{p}_a) \cdot \mathbf{n}] < 0,
\]
where $\mathbf{n} = (\mathbf{p}_b - \mathbf{p}_a) \times (\mathbf{p}_i - \mathbf{p}_a)$ is the normal of $\Pi$.

Since this holds for every internal face, all tetrahedra in $\mathcal{P}_{\text{before}}$ are non-overlapping.
\end{proof}

\subsection{Ray-Facet Intersection Predicates}
\label{appendix:ray-intersection}

This section provides the ray-facet intersection predicates used in curve tracking (Section~\ref{subsubsec:curve-tracking}).

\subsubsection{Ray-Edge Intersection (2D)}

Given a ray starting at $\mathbf{V} \in \mathbb{R}^2$ with direction $\mathbf{k}$, and an edge with endpoints $\mathbf{a}, \mathbf{b} \in \mathbb{R}^2$, we seek parameters $t$ and $u$ satisfying:
\begin{equation}
  \mathbf{V} + t\mathbf{k} = \mathbf{a} + u(\mathbf{b} - \mathbf{a}).
\end{equation}

Let $\mathbf{v}_1 = \mathbf{b} - \mathbf{a}$ and $\mathbf{v}_2 = \mathbf{a} - \mathbf{V}$. The determinant is:
\begin{equation}
  \Delta = \mathbf{k}_x \mathbf{v}_{1y} - \mathbf{k}_y \mathbf{v}_{1x}.
\end{equation}

If $\Delta \neq 0$:
\begin{equation}
  t = \frac{\mathbf{v}_{2x} \mathbf{v}_{1y} - \mathbf{v}_{2y} \mathbf{v}_{1x}}{\Delta}, \qquad u = \frac{\mathbf{k}_x \mathbf{v}_{2y} - \mathbf{k}_y \mathbf{v}_{2x}}{\Delta}.
\end{equation}

An intersection exists if $t > 0$ and $u \in [0, 1]$.

\subsubsection{Ray-Triangle Intersection (3D)}

Given a ray starting at $\mathbf{V} \in \mathbb{R}^3$ with direction $\mathbf{k}$, and a triangle with vertices $\mathbf{a}, \mathbf{b}, \mathbf{c} \in \mathbb{R}^3$, we compute the normal $\mathbf{n} = (\mathbf{b} - \mathbf{a}) \times (\mathbf{c} - \mathbf{a})$ and the ray-plane intersection:
\begin{equation}
  t = \frac{\mathbf{n} \cdot (\mathbf{a} - \mathbf{V})}{\mathbf{n} \cdot \mathbf{k}}.
\end{equation}

If $t > 0$, compute $\mathbf{p} = \mathbf{V} + t\mathbf{k}$ and its barycentric coordinates $(w_0, w_1, w_2)$ with respect to $(\mathbf{a}, \mathbf{b}, \mathbf{c})$. The point lies inside the triangle if $w_i \geq 0$ for all $i$.

\subsubsection{Implementation Notes}

For robustness in topology-preserving curve tracking, these predicates should be evaluated using exact arithmetic. We represent all coordinates as rational numbers and perform all arithmetic operations exactly. While this incurs computational overhead, it guarantees that intersection tests are consistent and prevents numerical errors from violating topological invariants.

In practice, we first attempt intersection tests using double-precision floating-point arithmetic. If the result is near-degenerate (e.g., $\Delta \approx 0$ in 2D, or barycentric coordinates near boundary values), we fall back to exact rational arithmetic to ensure correctness.

\end{document}